\documentclass[11pt]{article}
\pdfoutput=1
\usepackage{authblk}
\usepackage{palatino}
\usepackage{xcolor}
\usepackage{url}
\usepackage{subfigure}
\usepackage{amssymb}
\usepackage{amsmath}
\usepackage{graphicx}
\usepackage{boxedminipage}
\usepackage{xy}
\usepackage{url}

\usepackage{ifthen}
\usepackage{algorithm}
\usepackage{algpseudocode}

\xyoption{all}

\newtheorem{theorem}{Theorem}[section]
\newtheorem{lemma}[theorem]{Lemma}
\newtheorem{remark}[theorem]{Remark}
\newtheorem{claim}{Claim}[section]

\newtheorem{definition}{Definition}[section]
\newtheorem{corollary}[theorem]{Corollary}

\def\eod{\vrule height 6pt width 5pt depth 0pt}
\newenvironment{proof}{\noindent {\bf Proof:} \hspace{.4em}}
                      {\hspace*{\fill}{\eod}}
\newenvironment{proofof}[1]{\noindent {\bf Proof of #1:} \hspace{.4em}}
                           {\hspace*{\fill}{\eod}} 

\newcommand{\aut} {{\tt aut}}
\newcommand{\Ve}[1] {V_{#1}}
\newcommand{\Ed}[1] {E_{#1}}
\newcommand{\pair}[2] {\langle #1,#2\rangle}

\newcommand{\QB} {B}
\newcommand{\sig} {{\tt sig}}
\newcommand{\cnt} {{\tt cnt}}

\newcommand{\calC} {{\cal C}}
\newcommand{\calL} {{\cal L}}
\newcommand{\inpQ} {{\cal Q}}
\newcommand{\treeQ} {{\cal T}}
\newcommand{\sat} {{\tt Satellite}}
\newcommand{\SQ} {{\sf SQ}}

\newcommand{\val} {{\sf val}}
\newcommand{\PS} {{\sf PS}}
\newcommand{\DB} {{\sf DB}}
\newcommand{\Pp} {P^{+}}
\newcommand{\Pm} {P^{-}}
\newcommand{\cnth} {\cnt^*}
\newcommand{\hi} {{\rm hi}}
\newcommand{\Lbytwo} {\lfloor L/2 \rfloor}

\newcommand{\ceil}[1] {\lceil #1 \rceil}
\newcommand{\half}{\frac 12}

\newcommand{\eone} {p}
\newcommand{\etwo} {q}
\newcommand{\flower} {{\sf Flower}}
\newcommand{\eat}[1] {}
\newcommand{\PSE} {\PS}
\newcommand{\DBE} {\DB}
\newcommand{\NodeJoin} {\textbf{NodeJoin}}
\newcommand{\EdgeJoin} {\textbf{EdgeJoin}}

\newcommand{\e}{\epsilon}

\begin{document}
\title{Subgraph Counting: Color Coding Beyond Trees}
\author[1]{Venkatesan~T.~Chakaravarthy}
\author[2]{Michael~Kapralov}
\author[1]{Prakash~Murali}
\author[3]{Fabrizio~Petrini}
\author[3]{Xinyu~Que}
\author[1]{Yogish~Sabharwal}
\author[3]{Baruch~Schieber}
\affil[1,3]{IBM Research}
\affil[1]{\{vechakra, prakmura, ysabharwal\}@in.ibm.com}
\affil[3]{\{fpetrin, xque, sbar\}@us.ibm.com}
\affil[2]{EPFL}
\affil[2]{michael.kapralov@epfl.ch}
\maketitle
\begin{abstract}
  The problem of counting occurrences of query graphs in a large
  data graph, known as subgraph counting, is fundamental to
  several domains such as genomics and social network
  analysis. Many important special cases (e.g. triangle counting) have received significant attention.
  Color coding is a very general and powerful algorithmic technique for subgraph counting. 
  Color coding has been shown to be effective in several applications, but scalable implementations 
  are only known for the special case of {\em tree queries} (i.e. queries of treewidth one). 
  
  In this paper we present the first efficient distributed implementation for color coding that goes beyond tree queries: 
  our algorithm applies to any query graph of treewidth $2$. 
  Since tree queries can be solved in time linear in the size of the data graph, our contribution is the first step
  into the realm of colour coding for queries that require superlinear running time in the worst case.
  This superlinear complexity leads to significant load balancing problems on graphs with heavy tailed degree distributions.
  Our algorithm structures the computation to work around high degree nodes in the data graph, 
  and achieves very good runtime and scalability on a diverse collection of data and query graph pairs as a result. 
  We also provide theoretical analysis of our algorithmic  techniques, 
  showing asymptotic improvements in runtime on random graphs with 
  power law degree distributions, a popular model for real world graphs.
  \end{abstract}

  \section{Introduction}
Graphs serve as common abstractions for real world data, 
making graph mining primitives a critical tool for analyzing real-world
networks.  Counting the number of occurrences of a query graph in
a large data graph (subgraph counting, often referred to as motif counting) is an important problem with
applications in a variety of domains such as bioinformatics,
social sciences and spam detection (e.g.~\cite{irs, sociology-1, milo}).

Subgraph counting and its variants have received a lot of attention in the literature.
Substantial progress has been achieved for the case of small queries such as triangles or 4-vertex subgraphs:  
not only have very efficient algorithms been developed
(e.g.~\cite{triangleCohen, JhaSP15, triangle2, SuriV11}), but also
theoretical explanation of their performance on  popular graph models has been obtained ~(see \cite{BFMPSW14} and references therein).

Some of the recent work has addressed larger queries \cite{madduri,madduri4,BH13,graft,6413912}, 
but our understanding here is far from complete. Even for reasonably large graphs (a million edges) and small queries (e.g. $5$-cycles), 
the number of solutions tend to be enormous, running into billions.  
This explosion in the search space makes the subgraph counting problem very hard even 
for moderately large queries. Theoretically, the fastest known algorithm for counting occurrences of a 
$k$-vertex subgraph in an $n$-vertex data graph runs in time $n^{\omega k/3}$, 
where $O(n^\omega)$ is the time complexity of matrix multiplication (currently $\omega\approx 2.38$).
This improves upon the trivial algorithm with runtime $n^k$, 
but is prohibitively expensive even for moderate size queries.

\begin{figure}
\centering
\includegraphics[width=2in]{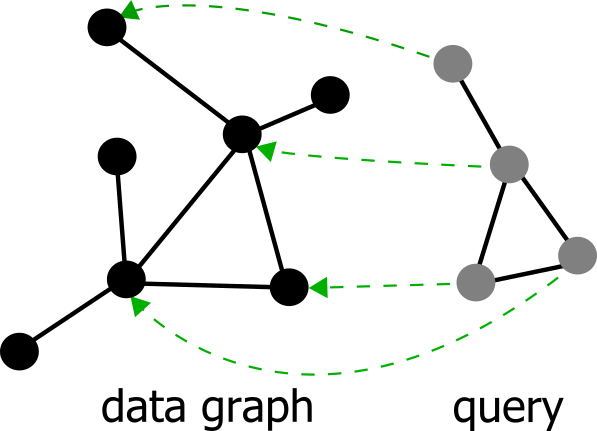}\hspace{1in}
\includegraphics[width=2in]{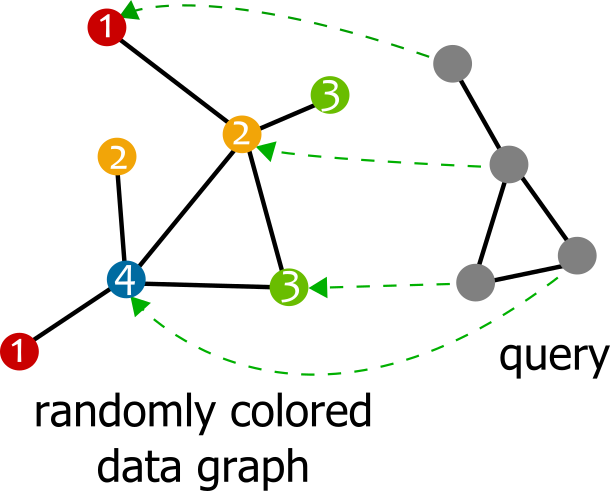}
\caption{Illustration of a match  (left) and a colorful match (right)}
\label{fig:match}
\end{figure}

To address the above issue, Alon et al.~\cite{AYZ95} proposed 
the {\em color coding} technique. Here, given a $k$-node query, 
we assign random colors between $1$ and $k$ to the vertices of the data graph, 
and count the number of occurrences of the query that are {\em colorful}, 
meaning the vertices matched to the query have distinct colors. See Figure \ref{fig:match}.
The count is scaled up appropriately to get an estimate on the actual number of occurances. 
The accuracy is then improved by repeating the process over multiple random colorings and taking the average.
Restricting the search to colorful matches leads to pruning of the search space and improved efficiency.
Using this method, Alon et al. obtained faster algorithms for cetain queries such as 
paths, cycles, trees and bounded treewidth graphs.

The power of color coding as a very general counting technique together with the importance of subgraph counting 
in various applications (as mentioned above) makes it important to design practically efficient and scalable implementations.
In a different work, Alon et al.~\cite{AlonDHHS08} applied the color coding technique
for counting the occurrences of {\em treelets} (tree queries) in biological networks. 
Color coding allowed them to handle tree queries up to size $10$ in protein interaction networks, 
extending beyond the reach of previously known approaches~\cite{PrzuljCJ04,HormozdiariBPS07, GrochowK07}. 
Recently, Slota and Madduri~\cite{SM13,madduri4} presented FASCIA, 
an efficient and scalable  distributed implementation of subgraph counting (via color coding), 
again for the case of  treelet queries. However, despite considerable interest in non-tree queries from several 
application domains (see the experimental section for details), 
the technque has not been explored for more general settings. 
In this work we present the first efficient distributed implementation of {\em color coding beyond tree queries}.

As part of their original color coding solution, Alon et al.~\cite{AYZ95} presented
faster algorithms for certain special classes of queries. 
They showed that if the query is a tree, then colorful subgraph counting 
can be solved in time $O(2^k m)$, i.e. in time {\em linear} in the size of the data graph. 
They extended the algorithm to show that if the query is close to a tree, 
specifically has {\em (small) treewidth $t$},  a running time of $O(2^k n^{t+1})$ can be achieved.  
Treewidth~\cite{tw} is a widely adopted measure of the intrinsic complexity of a graph.
Intuitively, it measures how close the topology of a given graph is to being a tree:
tree queries have treewidth $1$, and a cycle is the simplest example
of a treewidth 2 query. 
The above algorithm, restricted to trees, forms the basis
for the previously-mentioned treelet counting implementations \cite{SM13,madduri4,AlonDHHS08}.

While the runtime of the above algorithm is {\em linear} for the case of trees (i.e. {\em acyclic queries}), 
it becomes {\em at least quadratic} for query graphs of treewidth 2 and beyond. 
This phenomenon also manifests itself in practice:
on real world graphs with even moderately skewed degree distribution load imbalance is observed
and the running time tends to have quadratic dependence on the maximum degree of the graph.
Thus, even triangles (the smallest cyclic query) are harder to handle, and
have received considerable attention from the research community (as mentioned earlier).

The goal of this paper is to study the colorful subgraph counting problem on {\em queries of treewidth $2$}, 
taking the first step in the realm of color coding with cyclic queries. 
The class of queries of treewidth $2$ is quite rich. In particular, it contains all trees, cycles, series-parallel graphs and beyond. 
Figure~\ref{query-pic} shows treewidth $2$ queries (used in our experimental evaluation) drawn from real-world studies
on biological, social and collaboration networks \cite{Middendorf01032005,wu2012cwaunmcar,PhysRevE}.

To the best of our knowledge, the previously-mentioned algorithm
\cite{AlonDHHS08} 
is the best known algorithm for treewidth $2$ queries, and we use it as our baseline.
We rephrase this algorithm within our framework and devise a distributed implementation. 
The rephrased algorithm becomes a recursive procedure
that decomposes the query into simpler {\em path subqueries}, which
are then solved to get the overall count.
We thus refer to our baseline as the {\it Path Splitting algorithm ($\PS$)}.

\subsubsection*{Our Contributions} 
\indent\par{\bf 1.}~Building on the $\PS$ algorithm, we develop novel strategies 
that lead to significant performance gains in terms of runtime, 
scalability, and the size of graphs and queries handled.

\par{\bf 2.}~Our algorithm works by  decomposing the query to cycles and leaves, thereby reducing the problem of colorful subgraph counting on treewidth $2$ queries to 
counting (annotated) cycles.

\par{\bf 3.}~The decomposition in terms of cycles 
enables us to exploit the so-called degree ordering approach (e.g., \textsc{MinBucket} algorithm for triangle enumeration~\cite{BFMPSW14}) 
Specifically, we show how to force the computation process to (mostly) work around high degree vertices, 
leading to substantial speedups and scalability gains. 

\par{\bf 4.}~We present a detailed experimental evaluation of the algorithms on real-world graphs having
more than million edges and real-world queries of size up to 10 nodes. 
The results show that our strategies offer improvements of up to $28$x in terms of running time 
and exhibit improved scalability. 

\par{\bf 5.}~Finally, we complement our experimental evalutation by a theoretical 
analysis of the runtime of our degree ordering approach for cycle queries, on a popular class of 
random power law graphs (Chung-Lu graphs~\cite{Chung10122002}).
Our analysis provides justification for empirically observed performance gains of the approach.

\subsubsection*{Related Work}
Subgraph counting has received significant attention in the fields of 
computational biology~\cite{PrzuljCJ04,HormozdiariBPS07, GrochowK07}
and social network analysis~\cite{triangle,triangle1,triangle2,BFMPSW14, JhaSP15}. 
We give an overview of prior work on the problem (both theoretical and empirical) 
as well as techniques for making subgraph counting scalable, and explain how our contributions relate to this prior work.

{\it Color Coding and Approximate Subgraph Counting:} 
Color coding was introduced in an influential paper by Alon et al.~\cite{AYZ95} as a fast algorithm for finding
occurrences of a query in a data graph and counting the number of such occurrences.
In a different work, Alon et al.~\cite{AlonDHHS08} explored 
its applications to approximate subgraph counting (most commonly known as {\em motif counting}) in computational biology.
They were motivated by the fact that subgraph counting is an important primitive for characterizing biological networks~\cite{MilEtAl02}. 
Color coding allowed Alon et al. to count occurrences of treelets (tree queries) up to size $10$ in protein interaction networks, 
extending beyond the reach of previously known approaches~\cite{PrzuljCJ04,HormozdiariBPS07, GrochowK07}. 
A scalable distributed implementation of color coding for trees has been reported by Slota and Madduri~\cite{madduri, madduri4}, 
but no principled solutions beyond tree queries are known. 
ParSE~\cite{ZhaoKKM10} extends beyond tree queries, by considering 
query graphs that  can partitioned into subtemplates via edge cuts of size 1. 
However, the only class of query graphs that can be perfectly partitioned using this method is trees; 
ParSE resorts to brute force enumeration for other cases.
Our work provides the first principled approach to implementing color coding in a scalable way beyond trees queries. 
Further, our analysis of the runtime of our cycle counting subroutine on a random 
graphs with a power law degree distribution provides a theoretical justification of our algorithmic techniques.

While our work and the above-mentioned prior work \cite{AlonDHHS08,madduri,madduri4} 
count {\em non-induced} sugraphs, some other prior work~\cite{PrzuljCJ04,HormozdiariBPS07,GrochowK07} 
addressed the case of counting {\em induced} subgraphs.
The search space of non-induced subgraphs is larger and furthermore, 
these counts are more robust with respect to perturbations of the data graph~\cite{AlonDHHS08}. 

{\it Degree Based Approaches: }
Designing scalable subgraph counting algorithms turns out to be hard even for the simple case of triangle counting.
A naive approach lets each vertex enumerate pairs of neighbors and check if they are connected.
This leads to wasteful computations and also runs into load balancing issues on graphs with heavy tailed degree distributions~\cite{SuriV11}. 
The above issue has been addressed using a simple, but efficient solution (referred to as the \textsc{MinBucket} algorithm~\cite{triangleCohen, SuriV11}): 
each vertex  enumerates pairs of neighbors with degree no smaller than its own (with arbitrary tie breaking) and checks they are connected. 
It is not hard to see that this gives a correct count, and it has been empirically observed that this algorithm does not run 
into load balancing issues even on heavy tailed graphs~\cite{SuriV11}. The \textsc{MinBucket} heuristic has also been 
shown to give polynomial runtime improvement over the naive method when the input is a random graph with a power 
law degree distribution~\cite{BFMPSW14}. 
A recent work by Jha et al. \cite{JhaSP15} applies the degree based technique for coutning $4$-vertex queires.
There are a few prior approaches for arbitrary queries of ~\cite{6413912, AparicioRS14, graft},
but algorithms do not use degree information, and are comparable to the baseline algorithm used in our study.

To the best of our knowledge, prior to our work 
there has not been a systematic study of how \textsc{MinBucket} generalizes to larger subgraph counting 
problems. In this work we generalize the method for counting occurrences of treewidth 2 graphs, 
perform a thourough experimental evaluation and provide a theoretical runtime analysis of our technique in the random power law graph model.  
Our paper improves upon prior work along three axes: generality of queries handled, 
scalability of the proposed solution and theoretical analysis of the main algorithmic primitive on a class of graphs often used to 
model real world networks.

\section{Preliminaries}
\label{sec:prelims}
\noindent{\bf Subgraph counting problem.}
The {\em subgraph counting problem} is defined as follows. The input consists of a {\em query graph} $\inpQ=(\Ve{\inpQ}, \Ed{\inpQ})$ 
over a set of $k$ nodes and a data graph $G=(\Ve{G}, \Ed{G})$ over a set of $n$ vertices
and $m$ edges. The task is to count the number of (not necessarily induced) subgraphs of $G$ that are isomorphic to $\inpQ$. Formally, count the number of injective mappings $\pi:\Ve{\inpQ}\rightarrow \Ve{G}$ such that for any pair of query nodes $a, b\in \Ve{\inpQ}$,  if $\pair{q_1}{q_2} \in \Ed{\inpQ}$, then 
$\pair{\pi(q_1)}{\pi(q_2)} \in \Ed{G}$. We refer to such mappings $\pi$ as {\em matches}.

\noindent{\bf Color coding and colorful matches.} 
A {\em coloring} is a function $\chi:\Ve{G} \to \{1,2,\ldots,k\}$, where
for every vertex $u\in \Ve{G}$, $\chi(u)$ denotes its color. A match
$\pi$ from $V_{\inpQ}$ to $\Ve{G}$ is {\em colorful} if the 
vertices of $\inpQ$ are mapped to $k$ distinctly colored vertices in $G$, i.e.
$\bigcup_{a\in V_{\inpQ}} \chi(\pi(a))=\{1, 2, 3,\ldots, k\}$.
The main idea is that instead of counting all
possible matches of the $k$ vertices of the query graph  to the
vertices of the data graph, one first {\em colors} the vertices of the
data graph uniformly at random using $k$ colors, and then searches for
{\em colorful} matches.

\noindent{\bf Colorful subgraph counting problem.} In the {\em colorful subgraph counting problem} the task is to count the number of colorful matches of the query $\inpQ$ in $\Ve{G}$. 

Our setting counts the number of colorful matches or mappings from $\inpQ$ to the data vertices. Alternatively, we may want to count
the number of colorful subgraphs that are isomorphic to $\inpQ$. The latter quantity can be obtained by dividing the former by $\aut(\inpQ)$,
the number of automorphisms of $\inpQ$.
While it is computationally hard to compute $\aut(\inpQ)$ for an arbitrary query graph,
the quantity can be computed quickly for queries of relatively small size (say about $10$ nodes).
Given the above discussion, we focus on counting the number of colorful matches.

\noindent{\bf Treewidth.} 
Intuitively, if the query graph 
$\inpQ=(\Ve{\inpQ}, \Ed{\inpQ})$ has treewidth $t$ then $\inpQ$ can be
decomposed into subgraphs $Q_1,Q_2,\ldots$ such that
each subgraph $Q_i$ is also of treewidth $t$, and
each $Q_i$ has no more than $t$ nodes that belong also to
other subgraphs. We call such nodes the {\em boundary nodes} of $Q_i$.
In addition, the total number of distinct boundary nodes in all subgraphs 
$Q_1,Q_2,\ldots$ is at most $t+1$. 
Note that the decomposition can be done recursively
as each $Q_i$ has treewidth $t$,
until we are left only with subgraphs that have at most $t+1$ nodes.
This results in a {\em treewidth decomposition tree} denoted $\treeQ_\inpQ$.
A formal definition is givne below.

A tree decompsition of a query $|\inpQ|$ is a tree
$\treeQ = (\Ve{\treeQ}, \Ed{\treeQ})$, wherein each node $p\in \Ve{\treeQ}$ 
is associated with a subset of query nodes $S(p) \subseteq \Ve{\inpQ}$, called pieces, 
such that the following properties are true: 
(i) for every query edge $(a,b)\in \Ed{\inpQ}$, there exists a piece $S(p)$ (for some $p \in \Ve{\treeQ}$) 
that contains both $a$ and $b$; 
(ii) for every query node $a \in \Ve{\inpQ}$, 
the set of nodes whose pieces contain $a$ induce a connected subtree.
Alternatively, the second property states that if $a$ belongs to
pieces $S(p_1)$ and $S(p_2)$ for some $p_1$ and $p_2$, 
then $a$ must also belong to the piece $S(p)$ for any node $p$ found on the
(unique) path connecting $p_1$ and $p_2$ in $\treeQ$. The width of the
tree decomposition is the maximum cardinality overall pieces
minus one, i.e., $\max_p |S(p)|-1$. The treewidth $t$ of the query
is the minimum width over all its tree decompositions.

{\bf Approximate subgraph counting via color coding.} 
Counting the number of colorful matches turns out to be easier than counting the actual (not necessarily colorful) matches.
The price to pay is that the algorithm is randomized. We color the graph randomly and obtain the number of colorful matches,
and repeat the process independently at random a few times. Then, an estimate for the number of matches (occurances of the query) can be obtained
by taking the average.

For a given input graph $G$ and query $\inpQ$ let $n(G, \inpQ)$ denote the number of matches $\pi$ from $\inpQ$ to $G$. 
For a (random) coloring $\chi$ of vertices of $G$ let $n^{colorful}(G, \inpQ, \chi)$ denote the number of colorful matches of $\inpQ$ 
to $G$ under coloring $\chi$. It was shown~\cite{AYZ95,AlonDHHS08} that with proper normalization the colorful count 
$n^{colorful}(G, \inpQ, \chi)$ is an unbiased estimator of the actual count. Specifically,  the right normalization factor is $k^k/k!$, 
i.e. we have $(k^k/k!)\cdot {\bf E}_\chi[n^{colorful}(G, \inpQ, \chi)]=n(G, \inpQ)$. 
The variance of the estimator can also be bounded (see~\cite{AlonDHHS08}, section 2.1). 
Thus, taking the average of $n^{colorful}(G, \inpQ, \chi)$ under a few independently chosen colorings $\chi$ converges to 
the right answer, i.e. $n(G, \inpQ)$.  
Thus, in order to obtain an approximate subgraph counting algorithm it suffices to solve the colorful subgraph counting problem. 
The rest of the paper is devoted to designing a scalable solution to colorful subgraph counting.


\section{Overview}
\label{sec:overview}
The work of Alon et al.~\cite{AYZ95} yields a natural algorithm for the colorful subgraph counting problem
on bounded treewidth query graphs. This algorithm is based on the following intuition.
Suppose that we have found a colorful match $\pi$ for a subgraph $Q$ of the input query graph $\inpQ$,
and we wish to extend it into a colorful match $\pi'$ for $\inpQ$ by additionally fixing the mapping of the nodes outside $Q$. 
For this we do not need to know the mapping of the non-boundary nodes of $Q$, since they do not share edges with nodes outside $Q$.
Instead, it suffices to know the mapping of the boundary nodes (i.e., the nodes that share edges with nodes outside $Q$) 
and the set of colors used by $\pi$. The mapping of the boundary nodes is needed to ensure that 
for any edge from a boundary node to outside, the corresponding data vertices share an edge in the data graph;
and the set of colors is needed to avoid repeating a color already used by $\pi$. 
Analogously, in the setting of counting, in order to count the number of colorful matches for $\inpQ$,
we do not need a complete listing of colorful matches of $Q$.
Instead, we can group the colorful matches based on the set of colors used
and the mappings for the boundary nodes and it suffices to know the count per group.

Based on the above intuition, we apply dynamic programming to
count the number colorful matches of $\inpQ$. 
Let $\treeQ_{\inpQ}$ be the tree decomposition of $\inpQ$ with treewith $t$.
The algorithm processes $\treeQ_\inpQ$ in a bottom-up manner
and a creates a hash table (that we call a projection table) for each tree node.
The subgraph graph $Q$ associated with a node has at most $t$ boundary nodes 
and these nodes can be mapped to the data vertices in at most $n^t$ ways. 
In addition, we need to record the colors of the data vertices to which the nodes of $Q$ are
mapped. Since we focus on colorful matches,
the set of colors used (that we call ``signature") can be at most
${k \choose t}\leq 2^k$ (where $k$ is the size of the query graph).
For each combination of mappings to the boundary nodes and the signature,
we record the number of colorful matches of $Q$ consistent with the combination.
The number of entries in the table is at most $n^t 2^k$.
The projection table for a tree node can be computed from those of its children. 
We get the total number of colorful matches by performing an aggregation on the projection table of the root node.

Working in the realm of motif counting, Slota and Madduri \cite{madduri4} described an efficient distributed implementation of
the above algorithm for the case of tree queries and presented an experimental evaluation.
Trees have treewidth one hence, the size of projection tables is linear in the number of vertices
and the overall computation can be carried out in time linear in the graph size.
Our goal is to address a more general class of queries (beyond trees) in a distributed setting
and we focus on the case of queries of treewidth 2. 
Treewidth 2 queries are more challenging since in the worst case, 
the tables can be of size quadratic in the number vertices and the computation time also gets quadratic.

The construcion of our algorithm is motivated by the fact that
real life data graphs tend to exhibit variations in the degree distribution.
A naive implementation that treats all data vertices in
the same manner would result in
a lot of entries in the projection tables of the high degree vertices 
that do not lead to colorful matches for the overall input query.
Moreover, in a distributed setting the processors owning such vertices 
perform more computation leading to load imbalance.

Our algorithm is based on a crucial observation that any treewidth 2 query can be
recursively decomposed into (annotated) cycles or leaves. 
The core component of the algorithm is an efficient procedure for handling cycles
that employs a strategy based on degree based ordering of vertices.
This leads to reduction in wasteful computation, as well as improved load balancing. 
The procedure is inspired by a similar strategy used in prior work~\cite{BFMPSW14} for handling triangles.
The overall algorithm uses the above decomposition and the improved procedure for handling cycles.


\section{Overall Algorithm}
In this section we describe the overall structure of our subgraph counting algorithm 
that proceeds in two steps. 
In the {\bf first step}, we decompose the query into cycles and leaves (called blocks)
and construct a {\em decomposition tree} for the input query $\inpQ$
which is essentially a carefully chosen treewidth decomposition tree;
each node of the tree represents a block and encodes a convenient subquery. 
This step is independent of the data graph and can be viewed as a preprocessing phase for the query. 
Then in the {\bf second step} we traverse the tree in a bottom up manner, performing primitive counting operations 
over the data graph prescribed by the internal nodes and combining the results. 
The final count is produced by the root of the tree.

\subsection{Decomposition Tree}\label{sec:decomp-tree}
For an input query graph $\inpQ=(\Ve{\inpQ}, \Ed{\inpQ})$, construct the decomposition tree $T(\inpQ)$ by iteratively applying 
one of two primitive operations: contraction of a leaf edge or a cycle. 
As these operations are applied the number of nodes 
in the query $\inpQ$ decreases. 
At the same time new edges may appear in $\inpQ$ to represent contracted structures, 
and edges as well as nodes may get annotated with the identity of the contracted structures that they represent. 
Before defining the tree construction algorithm we need to introduce two definitions. 
First, we say that a cycle $\calC$ in $\inpQ$ is {\em contractible} if 
{\bf (a)} $\calC=(a_0, a_1,\ldots, a_{L-1})$ is induced 
(i.e. there are no  edges between nodes $a_0,a_1,\ldots, a_{L-1}$ except the edges of $\calC$) and 
{\bf (b)} cycle $\calC$ has most two boundary nodes (i.e., nodes that
share edges with nodes outside of $\calC$).  
Second, a {\em leaf edge} is an edge $\calL=(a,b)$, where $b$ is a leaf node (has degree one);
$a$ is called the boundary node of the leaf edge.
We use the common term {\em block} to refer to leaf edges and contractible cycles.

For example, consider the query named $\sat$ in Fig~\ref{fig:bd}.
The cycle $(i,j,k)$ is contractible with a single boundary node $i$,
the cycle $(a,b,c,d,e)$ is contractible with two boundary nodes $a$ and $c$,
and $(f,h)$ is a leaf edge. The cycle $(i,f,g)$ is not contractible since it has three boundary nodes.

We construct the decomposition tree $T(\inpQ)$ starting with an empty tree.
The tree is built bottom-up starting from the leaf level
and hence, the structure may be a forest with multiple roots in the intermediate stages.
Each iteration adds a new node and may make some of the existing roots as its children,
culminating in a tree. 

In the construction process we iteratively 
perform the following operations until $\inpQ$ contains a single node: 
find a block $B$ (a leaf edge or a contractible cycle) in $\inpQ$ and
remove it from $\inpQ$ (while possibly adding an edge to $\inpQ$),
and add a corresponding node  to $T(\inpQ)$. 
We iterate until $\inpQ$ contains a single node. We distinguish 3 cases.
%
\par\noindent{\bf Case 1:} 
{\it $B$ is a contractible cycle $\calC$ with exactly one boundary node $a\in \Ve{\inpQ}$:} 
Remove the nodes and edges of $\calC$ from $\inpQ$, except for node $a$. 
Erase any annotation found on $a$ in $\inpQ$ and annotate it with the block name $B$. 
\par\noindent{\bf Case 2:} 
{\it $B$ is a contractible cycle $\calC$ with two boundary nodes $a,b\in \Ve{\inpQ}$:} 
Remove the nodes and edges of $\calC$ from $\inpQ$, except for the nodes $a$ and $b$. 
Add an edge $(a,b)$ in $\inpQ$ and  annotate it with $B$.
Erase any annotation found on $a$ and $b$ in $\inpQ$.
\par\noindent{\bf Case 3:} 
{\it $B$ is a leaf edge $\calL=(a,b)$:} 
Remove $b$ and the edge from $\inpQ$. 
Erase any annotation found on node $a\in \inpQ$
and annotate it with the block name $B$.

The nodes and edges of $\QB$ inherit the annotations from $\inpQ$,
as they were before $\inpQ$ was transformed
(this ensures that the annotations on the boundary nodes that got erased get captured by the new annotation). 

Next we add a new node $\QB$ to the tree $T(\inpQ)$. 
If any node or edge in $\QB$ has an annotation $\QB'$,
make $\QB'$ a child of $\QB$ in $T(\inpQ)$.
This completes the construction of $T$. 
We show below that the process can find a block in each iteration and terminate successfully on every query of treewidth 2. 
Assuming termination, it is not difficult to see that the process produces a tree.
During contraction, every block $\QB'$ annotates a particular node or an edge of $\inpQ$,
recording the way in which it has been contracted. The annotation gets inherited by some other block $\QB$ in a subsequent iteration.
The block $\QB$ becomes the parent of $\QB'$. The annotation is erased in $\inpQ$, ensuring that no other block becomes a parent of $\QB'$.

\begin{figure}
\begin{center}
\begin{tabular}{c}
\includegraphics[width=5.5in]{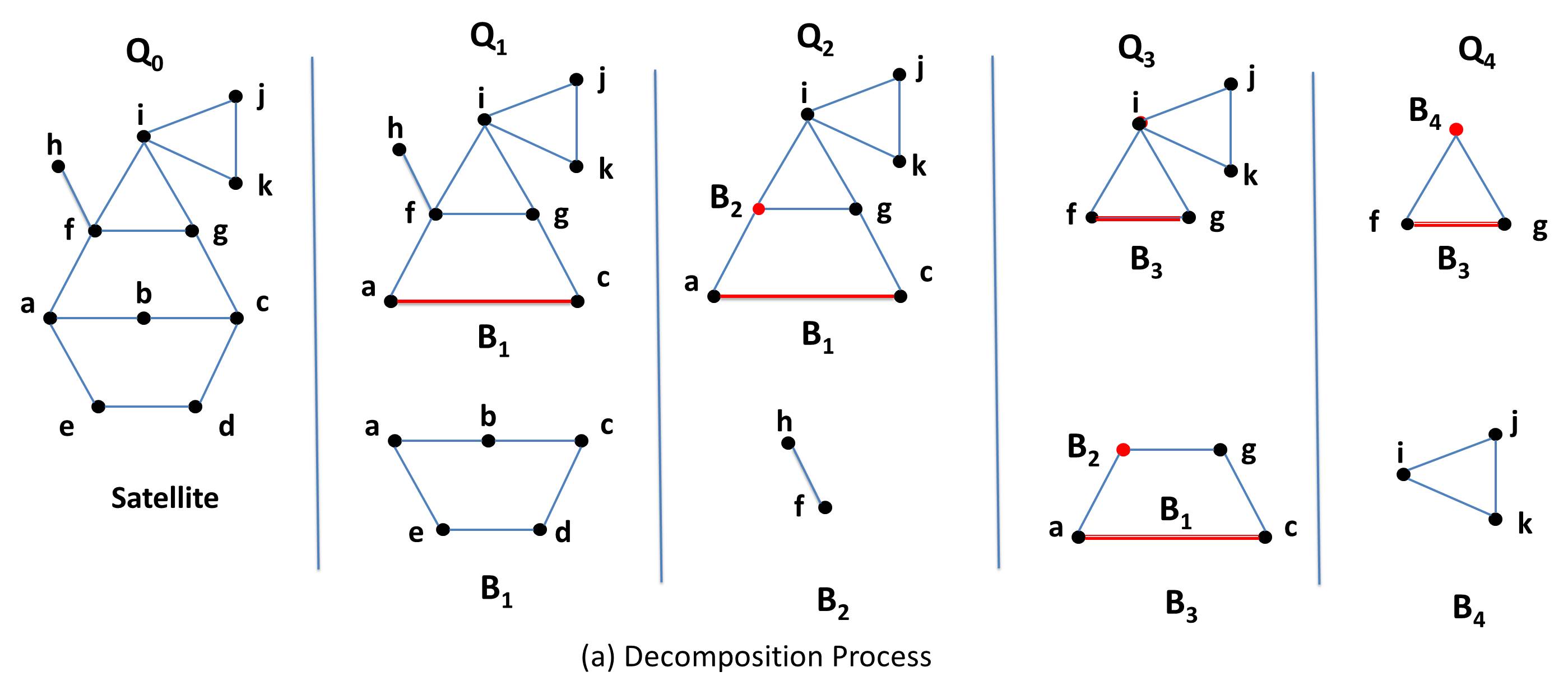} 
\\
\\
\begin{tabular}{ccc}
\includegraphics[width=2in]{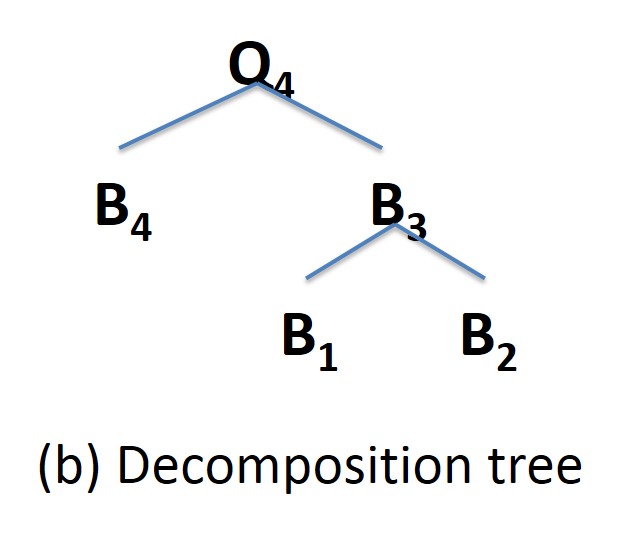}
&
\quad
\quad
&
\includegraphics[width=2in]{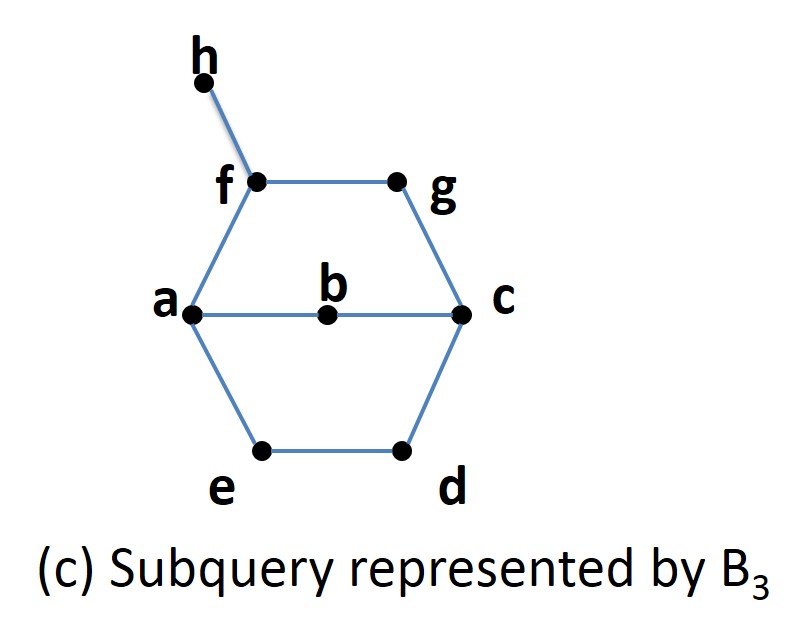}
\end{tabular}
\end{tabular}
\end{center}
\caption{Illustration of the decomposition process. the top row shows the sequence of queries considered in the process (the original query is on the left),  the bottom row shows the blocks that were contracted in each step.}
\label{fig:bd}
\end{figure}

Taking $\sat$ as the input query $\inpQ$, Figure~\ref{fig:bd} provides an illustration process, 
along with the output decomposition tree.
The bottom row shows the blocks being contracted and the top row shows the transformed $\inpQ$.
The first iteration contracts the cycle $\QB_1=(a,b,c,d,e)$.
A new edge $(a,c)$ is added to $\inpQ$, along with the annotation $\QB_1$, and $\QB_1$ is added to the tree.
The second iteration contracts the leaf block $\QB_2=(f,h)$.
Node $f$ is annotated as $\QB_2$ and the $\QB_2$ is added to the tree.
The third iteration contracts $\QB_3=(a,f,g,c)$, by adding an edge $(f,g)$ with the annotation $\QB_3$.
The block is added to the tree and it is made the parent of $\QB_1$ and $\QB_2$.
In the fourth iteration, the cycle $\QB_4=(i,j,k)$ is contracted.
Node $i$ gets annotated as $\QB_4$ and $\QB_4$ is added to the tree.
Finally, the query $Q_4$ is contracted leaving $\inpQ$ empty.
We add $Q_4$ as the root of the tree, making it the parent of $\QB_3$ and $\QB_4$. 

The following lemma guarantees that for any treewidth 2 query $\inpQ$,
the tree construction procedure will always find a block (a leaf edge
or a contractible cycle) in each iteration and terminate successfully.
The proof relies on prior work on nested ear decompositions
of treewidth 2 queries \cite{eppstein}.

\begin{lemma}
\label{lm:termination}
(i) Any treewidth 2 query $\inpQ$ contains a block; 
(ii) the transformed query resulting from the contraction process is
also a treewidth 2 query.
\end{lemma}
\begin{proof}
We first prove part (ii) of the lemma. If the contracted block has one
boundary node then no new edges are added to $\inpQ$, in which case
the tree $\treeQ_\inpQ$ for the updated $\inpQ$ is given by deleting all 
the nodes not in the updated $\Ve{\inpQ}$ from the subsets $S_\inpQ (t)$.
If the contracted block has two boundary nodes $a$ and $b$ then the
edge $(a,b)$ is added to $\inpQ$. In this case we get the tree for 
the updated $\inpQ$ by replacing each occurrence of
the nodes not in the updated $\Ve{\inpQ}$ by $b$. Note that the size of each subset is still
at most 3, nodes associated with subsets that contain $b$ form a
connected component, and for at least one subset $S_\inpQ(t)$,
$\{a,b\} \subseteq S_\inpQ(t)$.    

We now prove part (i). First, Root the tree $\treeQ_\inpQ$ at an arbitrary
non-leaf node. This induces an ancestor-descendant relationship on the
nodes in $\Ve{\treeQ}$. Note that if there are two nodes
$\{t,t'\} \subseteq \Ve{\treeQ}$, such that 
$S_\inpQ(t') \subseteq S_\inpQ(t)$, node $t'$ can be omitted and
all its children connected to $t$. Thus from now on we assume that
no subset $S_\inpQ(t)$ is contained (or identical) to another subset.

We need the following definition and claim.
\begin{definition}
For a node $t \in \Ve{\treeQ}$, let $\inpQ_t$ be the subgraph  of
$\inpQ$ induced by the
nodes that are in the union of the subsets associated with
the nodes of $\treeQ_\inpQ$ in the subtree rooted at $t$.
\end{definition}

\begin{claim}\label{cl:block}
For every node $t \in \Ve{\treeQ}$, 
either $\inpQ_t$ contains a block, or $\inpQ_t$ is
a path whose endpoints are in the subset
associated with the parent of $t$ (if such exists).
\end{claim}
Before proving the claim we show how it implies the lemma. Since the
claim holds also for the root of $\treeQ_\inpQ$ then either $\inpQ$
contains a block or it is a path in which case it also contains a leaf
block.
\end{proof}

\begin{proofof}{Claim~\ref{cl:block}}
We prove the claim by induction. The base of the induction is
a leaf node.
Consider a leaf node $t \in \Ve{\treeQ}$. 
There are two possibilities: (i)  $S_\inpQ(t) = \{x,y\}$, and 
(ii) $S_\inpQ(t) = \{x,y,z\}$.
If $S_\inpQ(t) = \{x,y\}$, then at least one node, say $y$, is only
connected to $x$ and thus $(x,y)$ is a leaf edge.

If $S_\inpQ(t) = \{x,y,z\}$, then consider the subgraph induced by
$\{x,y,z\}$. If this subgraph is a triangle then it must be a
contractible cycle.
The only remaining case is the subgraph induced by
$\{x,y,z\}$ forms a path. Assume that the endpoints of this path are
$x$ and $z$. If one of these endpoints, say $z$, is not in the subset 
associated with the parent of $t$ then $(y,z)$ is a leaf edge.
Otherwise, let $t'$ be the parent of $t$, we have  
$S_\inpQ(t) \cap  S_\inpQ(t') = \{x,z\}$.

For the inductive step consider a non-leaf node $t \in \Ve{\treeQ}$.
If $\inpQ_{t'}$ for any child $t'$ of $t$ contains a block then we are
done. Assume that this is not the case. Consider first the case that
$t$ has a single child $t'$. By the inductive hypothesis $\inpQ_{t'}$
is a path whose endpoints $x$ and $y$ are in $S_\inpQ(t)$.
Let $S_\inpQ(t) = \{x,y,z\}$. If $z$ is connected to both $x$ and $y$
then the cycle closed by $z$ is a contractible cycle.
If $z$ is connected to only one endpoint, say $y$, then we get a path
with endpoints $x$ and $z$. If either $x$ or $z$ are not in the subset 
associated with the parent of $t$, then the missing endpoint is leaf
node.  If both $x$ and $z$ are in the subset 
associated with the parent of $t$ then the inductive claim follows.

Next, Consider the case that $t$ has several children. If two of the
children of $t$, say $t'$ and $t''$, share endpoints then the cycle
formed by  $\inpQ_{t'}$ and  $\inpQ_{t''}$ is contractible.
Otherwise, $t$ must have exactly two children, say $t'$ and $t''$,
with endpoint $\{x,y\}$ and $\{y,z\}$, forming a path with endpoints
$x$ and $z$. If $z$ is connected also to $x$ 
then the cycle closed by the edge $(x,z)$ is a contractible cycle.
If either $x$ or $z$ are not in the subset 
associated with the parent of $t$, then the missing endpoint is a leaf
node. If both $x$ and $z$ are in the subset 
associated with the parent of $t$ then the inductive claim follows.
\end{proofof}

An input query may admit multiple decomposition trees 
and the choice of the tree influences the performance of our algorithm.
In Section \ref{sec:planner}, we present a heuristic for finding a good decomposition.
Each node of the tree represents a block and it will be 
convenient to view to the node simply as the block represented by it. 

At this point, it is interesting to consider tree queries studied by Slota and Madduri \cite{madduri4}.
Given a tree query, their algorithm fixes a suitable query node as the root and iteratively processes the tree in a bottom-up manner. 
The algorithm implicitly uses a decomposition tree.
However, since trees do not have cycles, the decomposition tree consists of only leaf edge blocks.
In contrast, the decomposition trees of treewidth two queries involve the more challenging case of cycles as well.

\subsection{Tree Traversal}\label{sec:traversal}
Here, we describe the second step of the algorithm that traverses the decomposition tree in a bottom-up manner
and computes the number of colorful matches of the blocks in the data graph. For this purpose, 
we define the notion of subqueries represented by blocks.

A {\em subquery} $Q$ of the input query $\inpQ$ refers to any induced subgraph of $\inpQ$.
Consider a block $B$ and let $U$ be the union of nodes found in the block $B$ and its descendant blocks in the tree.
The {\em subquery represented by $B$}, denoted $\SQ(B)$, refers to the subquery induced by $U$.
For example, Figure~\ref{fig:bd} shows the subquery represented by the block $\QB_4$.
The decomposition tree yields a nested hierarchy of subqueries: the root block represents the whole input query
and for any block $B$ with the parent $B'$, the subquery $\SQ(B)$ is contained within $\SQ(B')$.

Let $B$ be a block. A node $a\in \SQ(B)$ is said to be a {\em boundary node},
if $a$ shares an edge with a node outside $\SQ(B)$. It is not hard to see that these boundary nodes
are the same as the boundary nodes of $B$ (identified during the tree construction process).
Thus, $\SQ(B)$ can have at most two boundary nodes.

Before describing the counting algorithm we extend the notion of colorful matches to subqueries:
a colorful match for a subquery $Q=(\Ve{Q},\Ed{Q})$ is an injective mapping $\pi:\Ve{Q}\rightarrow \Ve{G}$, 
such that for any edge $(a,b)\in \Ed{Q}$, $(\pi(a), \pi(b))\in \Ed{G}$,
and the vertices of $Q$ are mapped to distinctly colored vertices of $G$. 

The algorithm traverses the tree in a bottom-up manner.
For each block $B$, it outputs a succinct synopsis of the set of colorful matches of 
the subquery $\SQ(B)$, using a projection table and signature
(as outlined in Section~\ref{sec:overview}).
that we now define precisely.


{\it Signature: }
Let $K=\{1,2,\ldots,k\}$ denote the set of colors used in the data graph, where $k$ is the size of
the input query $\inpQ$. The term {\em signature} refers to any subset $\alpha\subseteq K$.
For a subquery $Q$ and a colorful match $\pi$ of $Q$, the {\em signature} of $\pi$ refers
to the set of colors of the data vertices used by $\pi$ and it is denoted $\sig(\pi)$,
i.e., $\sig(\pi)=\cup_{a\in Q}\{\chi(\pi(a))\}$. 

{\it Projection Tables: }
Let $Q$ be subquery with two boundary nodes $a$ and $b$.
For a pair of data vertices $u$ and $v$ and  a signature $\alpha\subseteq K$
let $\cnt(u,v,\alpha|Q)$ denote the number of colorful matches of $Q$ wherein the 
boundary nodes $a$ and $b$ are mapped to $u$ and $v$ and the signature of $\pi$ is $\alpha$:
\begin{align*}
    \cnt(u, v, \alpha|Q) = |\{\pi\in \Pi~:~\pi(a) = u &\mbox{ and }\pi(b) = v \\
                                           &\mbox{ and }\sig(\pi) =\alpha\}|,
\end{align*}
where $\Pi$ is the set of all the colorful matches of $Q$.

These counts can be conveniently represented in the 
form a hash table with $(u,v,\alpha)$ forming the key and the count forming the value.
We refer to any encoding of the above counts
(such as the hash table above) as the {\em projection table} of $Q$.
In the worst case, the table may have size quadratic in the input data graph.
However, a significant fraction of the triplets will have a count of zero and we maintain only the non-zero counts.

The projection table for subqueries having a single boundary node $a$ is defined in a similar manner.
For a data vertex $u$ and a signature $\alpha\subseteq K$, define
\[
\cnt(u,s|Q) = |\{\pi\in \Pi~:~\pi(q) = u \mbox{ and } \sig(\pi) =\alpha\}|.
\]

\begin{figure}[t]
\begin{center}
\begin{boxedminipage}{\hsize}
\begin{small}
\begin{tabbing}
xx\=xx\=xx\=xx\=xxx\=xxx\=\kill
\textbf{Overall Algorithm}\\
\> 1. Compute a decomposition tree $T(\inpQ)$ for the input query $\inpQ$.\\
\> 2. Traverse the tree bottom-up.\\
\> \> For each non-root block $B$:\\
\> \> \> Use the projection tables of the children blocks of $B$ and\\
\> \> \> \> compute the projection table for $B$\\
\> 3. Output the number of colorful mathes of the subquery \\
\>\>\> \> represented by the root-block.
\end{tabbing}
\end{small}
\end{boxedminipage}
\end{center}
\caption{Overall Algorithm}
\label{fig:overall}
\end{figure}

\subsection{Computing the Counts}
Given a decomposition tree, the algorithm 
works based on the fact that the projection table for a block can be 
computed by joining the projection table of its children blocks.

As an illustration of the idea, consider the block $B_3$ having boundary nodes $f$ and $g$, 
and the subquery represented by it (Figure \ref{fig:bd}).
For a pair of vertices $u$ and $v$, and a signature $\alpha$,
the projection count $\cnt(u,v,\alpha|B_3)$ can be computed as follows.
The block consists of the path $(a,f,g,c)$, and any match $\pi$ for the subquery 
must map these nodes to vertices $(x,u,v,y)$ that form a path in the data graph.
The block is annotated by its children blocks $B_1$ with boundary nodes $a$ and $c$, 
and $B_2$ with boundary node $f$.
Any pair of matches $\pi_1$ and $\pi_2$ for $\SQ(B_1)$ and $\SQ(B_2)$ can be extended as matches for $\SQ(B_3)$,
as long as their signatures $\alpha_1$ and $\alpha_2$ are disjoint (since the blocks do not share any node)
and are contained within $\alpha$.
Therefore, we can derive the desired count by performing the following aggregation
over all quadruples $(x,y,\alpha_1,\alpha_2)$ satisfying the properties:
$(x,u,v,y)$ forms a path in the data graph; $\alpha_1,\alpha_2\subseteq \alpha$; 
$(\alpha_1\cap \alpha_2)$ is empty. 
The aggregation is:
\[
\cnt(u,v,\alpha|B_3) = \sum_{x,y}\sum_{\alpha_1,\alpha_2} \cnt(x,y, \alpha_1|B_1) \times \cnt(u, \alpha_2|B_2).
\]

We can express the projection counts for any block in the above manner.
However, as the number of children increases, the cartesian product involved in the aggregation
would be prohibitively expensive. 
Our procedures efficiently simulate the aggregation
by performing a sequence of join operations involving the projection tables of children blocks.
	
Given a decomposition tree, 
the algorithm traverses the decomposition tree in a bottom-up manner, 
computing the projection tables for all the blocks
and culminates in the root-block representing the whole input query. 
At this step, instead of producing a projection table,
the algorithm simply computes the number of colorful matches.
The pseudo-code is shown in Figure~\ref{fig:overall}.

\section{Solving Blocks}\label{sec:solving-blocks}
The main step of the algorithm is the construction of the
projection tables of a block from its children blocks.
In this section we develop efficient procedures for handling cycles.
For the sake of highlighting the main ideas,
we first focus on the case of cycles found at a leaf level of the decomposition tree
(such as the cycle $\QB_1$ in Figure~\ref{fig:bd});
these cycles do not have other blocks annotating them.
General cycles are handled by extending these ideas as discussed later.

\subsection{Solving Cycles at the Leaf Level} 
Consider a cycle block $\calC=(a_0, \ldots, a_{L-1})$ of length $L$ without annotations.
The cycle may have at most two boundary nodes. 
We discuss the more interesting case where the number of boundary nodes is exactly two; 
the other cases are handled in a similar fashion.
Let the two boundary nodes of the cycle be $a_{\eone}$ and $a_{\etwo}$, for some $0\leq \eone, \etwo\leq L-1$.
We present two procedures for computing the projection table of $\calC$: a baseline procedure that uses 
a path splitting strategy and an efficient procedure guided by a degree based ordering of vertices.

\begin{figure}[t]
\begin{center}
\begin{boxedminipage}{\hsize}
\begin{small}
\begin{tabbing}
xx\=xx\=xx\=xx\=xxx\=xxx\=\kill
\textbf{Procedure 1: Computing Projection Table for $\Pp$}\\
\> For each edge $(u,v)$ in the data graph $G$\\
\> \> $\cnt(u,v,\alpha|\Pp_{\eone,\eone\oplus 1}) \leftarrow 1$, where $\alpha= \{\chi(u), \chi(v)\}$.\\
\> For $j=\eone\oplus 2, \eone \oplus 3, \ldots, \etwo$\\
\> \> For each triple $(u,v,\alpha)$ with $\cnt(u,v,\alpha|\Pp_{\eone,j\ominus 1})\neq 0$\\
\> \> \> For each edge $(v, w)$ in $G$ such that $\chi(w)\not\in \alpha$ do:\\
\> \> \> \> Let $\alpha'=\alpha \cup \{\chi(w)\}$.\\
\> \> \> \> Increment $\cnt(u,w,\alpha'|\Pp_{\eone,j})$ by $\cnt(u,v,\alpha|\Pp_{\eone,j\ominus 1})$.\\
\\
\textbf{Procedure 2: Computing Projection Table for $\calC$}\\
\> For each entry $(u, v, \alpha_1)$ with $\cnt(u,v,\alpha_1|\Pp)\neq 0$\\
\> \> For each entry $(u,v,\alpha_2)$ with $\cnt(u,v,\alpha_2|\Pm)\neq 0$\\
\> \> \> If $\alpha_1 \cap \alpha_2 = \{\chi(u),\chi(v)\}$\\
\> \> \> \> $\alpha' \leftarrow \alpha_1 \cup \alpha_2$\\
\> \> \> \> $\val_1 \leftarrow \cnt(u,v,\alpha_1|\Pp)$; \quad $\val_2 \leftarrow \cnt(u,v,\alpha_2|\Pm)$\\
\> \> \> \> Increment $\cnt(u,v,\alpha'|\calC)$ by $\val_1 \times \val_2$.
\end{tabbing}
\end{small}
\end{boxedminipage}
\end{center}
\caption{$\PS$ Algorithm}
\label{fig:path-split}
\end{figure}

\noindent{\bf Path Splitting Algorithm ($\PS$).}
For two nodes $a_s$ and $a_t$ on the cycle, let $\Pp_{s,t}$ and $\Pm_{s,t}$ be the paths obtained by traversing the cycle
from $a_s$ to $a_t$ in the clockwise and counter-clockwise directions, respectively, i.e.,
$\Pp_{s,t}= (a_s, a_{s\oplus 1}, \ldots, a_t)$ and $\Pm_{s,t} = (a_s, a_{s\ominus 1}, \ldots, a_t)$,
where $\oplus$ and $\ominus$ refer to addition and subtraction modulo $L$. 

Let $\cnt(\cdot,\cdot,\cdot|\Pp_{s,t})$ denote 
the projection counts for path $\Pp_{s,t}$ taking $a_s$ and $a_t$ as the boundary nodes. 
Namely, for a triple $(u,v,\alpha)$, let $\cnt(u,v,\alpha|\Pp_{s,t})$ 
denote the number of colorful matches for $\Pp_{s,t}$
wherein $\pi(a_s)=u$, $\pi(a_t)=v$ and $\sig(\pi)=\alpha$.
A similar notion is defined for the paths $\Pm_{s,t}$.

\begin{figure}
\begin{center}
\includegraphics[width=4.0in]{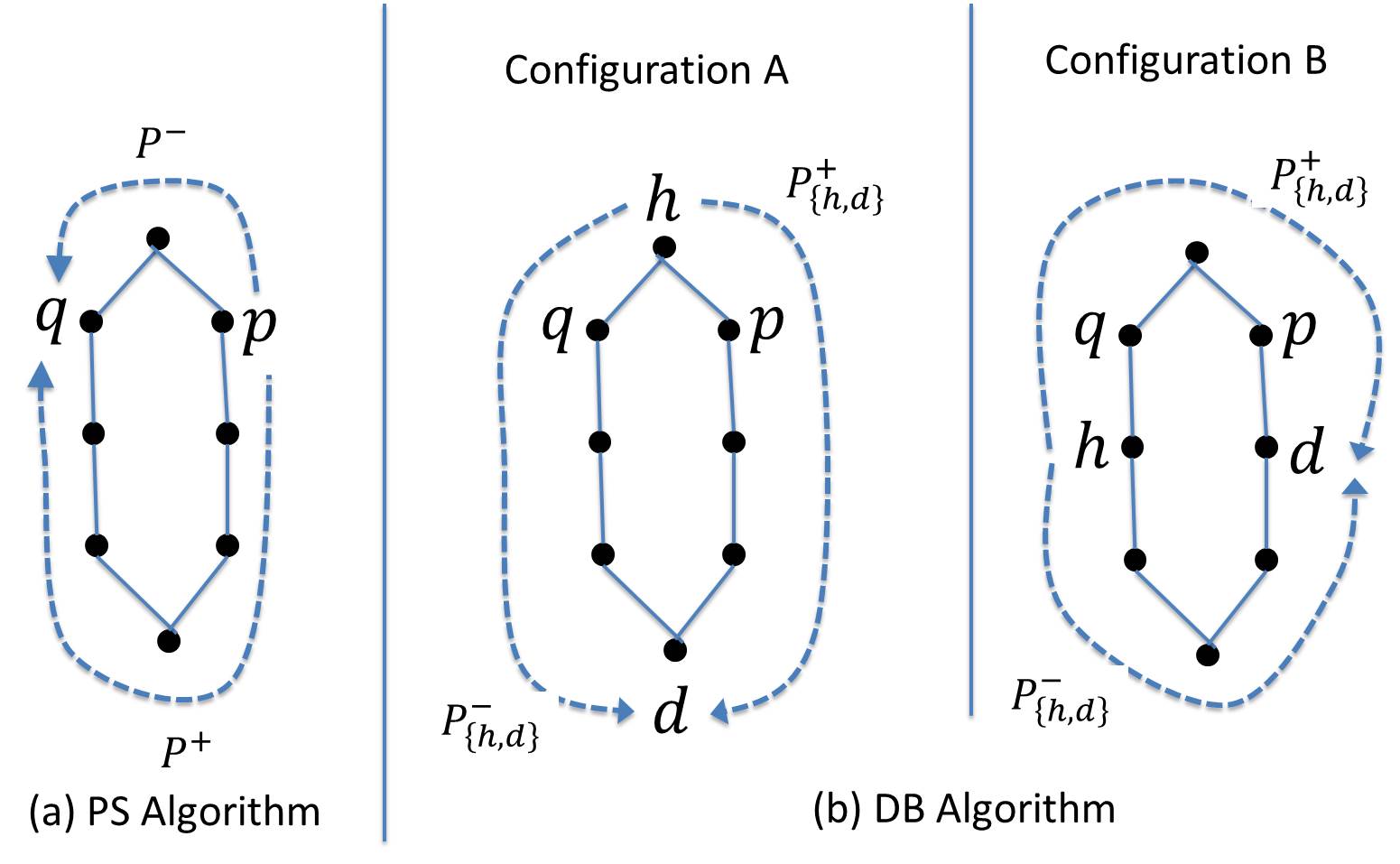}
\end{center}
\caption{$\PS$ and $\DB$ Illustrations.}
\label{fig:confy}
\end{figure}

The procedure splits the cycle into two paths along the boundary nodes,
given by $\Pp_{\eone,\etwo}$ and $\Pm_{\eone,\etwo}$;
we refer to these special paths as $\Pp$ and $\Pm$.
See Fig~\ref{fig:confy} (a) for an illustration.

The projection table for $\Pp$ is constructed iteratively,
by building the tables for the paths $\Pp_{\eone,j}$, for each node $a_j$ found along the path.
This is accomplished by extending the projection table for the prior path $\Pp_{\eone, j\ominus 1}$ via a join with the edges of the data graph.
The pseudocode is given in Figure~\ref{fig:path-split} (Procedure 1).
We assume that all the counts are initialized to zero.
The first iteration is handled by directly reading the edges of the data graph.
In the subsequent iterations, we extend every triple $(u,v,\alpha)$ with non-zero count $\cnt(u,v,s|\Pp_{\eone,j\ominus 1})$,
with any edge $(v,w)$, provided the resulting match is colorful.
The counts for $\Pm$ are constructed analogously.
Finally, the projection table for the cycle $\calC$ is obtained by joining the 
counts of $\Pp$ and $\Pm$, as shown in Procedure 2.
Here, a pair of triples $(u,v,\alpha_1)$ and $(u,v,\alpha_2)$ are joined, if the resulting match is colorful.

\noindent{\bf Discussion of baseline.}
As discussed below (Section \ref{sec:gen-cycle}), the $\PS$ procedure can be extended to
handle general cycles with annotations, and yields an algorithm for
handling treewidth 2 queries.
The resultant $\PS$ algorithm is equivalent to the original color coding algorithm of Alon et al. \cite{AYZ95}. 
Prior work \cite{madduri4,AlonDHHS08} on colorful subgraph counting utilize the algorithm of Alon et al. as the basis 
for counting tree queries (treelets). We developed a distributed implementation of the $\PS$ algorithm,
and use it as the baseline in our experimental study.
Known techniques for subgraph counting with large queries (e.g. \cite{6413912, graft}) 
employ similar graph traversal techniques, making {\PS} consistent with the state of the art 
for subgraph counting as well as color coding.

We develop an procedure, called Degree Based (\DB) algorithm, 
that outperforms the $\PS$ algorithm for practical graphs and queries. It is motivated by the following observations.
First, the paths $\Pp$ and $\Pm$ may have uneven lengths (for instance, in Figure \ref{fig:confy}), 
$|\Pp| = 6$ and $|\Pm|=2$) and the processing of the longer path dominates the overall running time.
Second, in real-graphs with skewed degree distributions, 
high degree vertices tend to have more paths passing through them,
which populate the projection tables of $\Pp$ and $\Pm$.
However, significant fraction of these paths do not find appropriate 
counterparts in the other table to complete a match, leading to wasteful computations.
Third, in a distributed setting, the above phenomenon manifests 
as higher load on processors owning high degree vertices, leading to load imbalance.

It is not difficult to address the first issue alone. The only intricacy is that 
when the paths are split evenly, the boundary nodes may appear internally on the 
the paths (see Figure \ref{fig:confy} with a split across nodes denoted $h$ and $d$).
This can be handled by recording the mapping for the boundary nodes as part of the 
projection counts. We implemented the above algorithm as well
and noticed that the issue of wasteful computations and load imbalance still persists.
And furthermore, performance of the $\PS$ algorithm and the modified implementations does not differ
significantly on our benchmark graphs and queries.

\noindent{\bf Degree Based Algorithm $(\DB$).}
The $\DB$ algorithm addresses all the three issues 
by using the strategy of building the paths from high degree vertices.

Arrange the data vertices in the increasing order of their degree;
if two vertices have the same degree, 
the tie is broken arbitrarily, say by placing the vertex having the least id first. 
We say that a vertex $u$ is {\em higher} than a vertex $v$,
if $u$ appears after $v$ in the above ordering and this is denoted ``$u \succ v$''.

Consider the input cycle $\calC=(a_0, a_1, \ldots, a_{L-1})$ with boundary nodes $a_{\eone}$ and $a_{\etwo}$
and let $\pi$ be a colorful match for $\calC$ 
that maps the above nodes to data vertices $u_0, u_1, \ldots, u_{L-1}$, respectively. 
Among these data vertices, let $u_j$ be the highest vertex. 
We refer to the corresponding node $a_j$ as the {\em highest} node of $\pi$.

The idea is to partition the set of colorful matches into $L$ groups based on their highest node $a_h$
and compute the projection table for each group separately. 
For a pair of data vertices $u$ and $v$, and a signature $\alpha$,
let $\cnt(u,v,\alpha|\calC, \hi=h)$ denote the number of colorful matches of $\pi$ for $\calC$,
wherein $\pi(a_{\eone})=u$, $\pi(a_{\etwo})=v$, $\sig(\pi)=\alpha$ and $a_h$ is the highest node of $\pi$.
The projection table for $\calC$ can be obtained by aggregating the above counts:
for any triple $(u,v,\alpha)$,
\begin{eqnarray}
\cnt(u,v,\alpha|\calC) &=& \sum_{h=0}^{L-1}\cnt(u,v,\alpha|\calC,\hi=h).
\label{eqn:aggr}
\end{eqnarray}

We next describe an efficient procedure for computing the counts $\cnt(u,v,\alpha|\calC,\hi=h)$.
The concept of high starting matches plays a crucial role in the procedure.
Let $a_d$ be the node diagonally opposite to $a_h$ on the cycle, i.e., $d = h \oplus \Lbytwo$. 
The procedure splits the cycles into two paths $\Pp_{h,d}$ and $\Pm_{h,d}$;
Figure~\ref{fig:confy} (b) shows the paths for two sample values of $h$.
Let $a_j$ be a node found on the path $\Pp_{h,d}$,
A colorful match $\pi$ for $\Pp_{h,j}$ is said to be {\em high-starting},
if the data vertex $\pi(a_h)$ is higher than all the other data vertices used by $\pi$,
i.e., $\pi(a_h) \succ \pi(a_i)$, for all nodes $a_i$ on the path $\Pp_{h,j}$.
For a pair of vertices $u$ and $v$, and a signature $\alpha$,
let $\cnth(u,v,\alpha|\Pp_{h,j})$ denote the number of high-starting 
colorful matches for the path $\Pp_{h,j}$ wherein $\pi(a_h)=u$, $\pi(a_j)=v$ and $\sig(\pi)=\alpha$.

We then count the high-starting colorful matches for the two paths,
which can be accomplished via edge extensions, as in the $\PS$ algorithm.
However, the current setting offers a crucial advantage:
we can dictate that the starting node $a_h$ is the highest node,
meaning whenever an entry $(u,v,\alpha)$ gets extended by an edge $(v, w)$,
we can impose the condition that $u$ is higher than $w$ in the degree based ordering.
Imposing the condition leads to a significant pruning of the tables.
The pseudo-code is given in Figure~\ref{fig:db-algo} (Procedure 1).

\begin{figure}[t]
\begin{center}
\begin{boxedminipage}{\hsize}
\begin{small}
\begin{tabbing}
xx\=xx\=xx\=xx\=xxx\=xxx\=\kill
\textbf{Procedure 1: Compute $\cnth(u,v,\alpha|\Pp_{h,d})$}\\
For each edge $(u,v)$ in the data graph $G$ with $u \succ v$\\
\> $\cnth(u,v,\alpha|\Pp_{h,h\oplus 1}) \leftarrow 1$, where $\alpha= \{\chi(u), \chi(v)\}$.\\
For $j=h\oplus 2, a\oplus 3, \ldots, d$\\
\> For each triple $(u,v,\alpha)$ with $\cnth(u,v,\alpha|\Pp_{h,j\ominus 1})\neq 0$\\
\> \> For each edge $(v, w)$ in $G$ s.t. $u \succ w$ and $\chi(w)\not\in \alpha$:\\
\> \> \> Let $\alpha'=\alpha \cup \{\chi(w)\}$.\\
\> \> \> Incr. $\cnth(u,w,\alpha'|\Pp_{h,j})$ by $\cnth(u,v,\alpha|\Pp_{h,j\ominus 1})$.\\
\\
\textbf{Procedure 2: Compute $\cnth(x,y,\alpha|\calC,\hi=h)$ for Config. (A)}\\
For each entry $(u, v, x, \alpha_1)$ with $\cnth(u,v,x,\alpha_1|\Pp_{h,d})\neq 0$\\
\> For each entry $(u,v,y,\alpha_2)$ with $\cnth(u,v,y,\alpha_2|\Pm_{h,d})\neq 0$\\
\> \> If $\alpha_1 \cap \alpha_2 = \{\chi(u),\chi(v)\}$\\
\> \> \> $\alpha' \leftarrow \alpha_1 \cup \alpha_2$\\
\> \> \> $\val_1 \leftarrow \cnth(u,v,x,\alpha_1|\Pp_{h,d})$; \\
\> \> \> $\val_2 \leftarrow \cnth(u,v,y,\alpha_2|\Pm_{h,d})$\\
\> \> \> Incr. $\cnth(x,y,\alpha'|\calC, \hi=h)$ by $\val_1 \times \val_2$.
\end{tabbing}
\end{small}
\end{boxedminipage}
\end{center}
\caption{$\DB$ Algorithm}
\label{fig:db-algo}
\end{figure}

While the degree based strategy is more efficient,
we need to address an intricacy regarding the projection aspects.
In contrast to the $\PS$ algorithm, the $\DB$ algorithm 
splits at the highest node and consequently, the boundary nodes $\eone$ and $\etwo$ may appear inside the paths.
Thus, in order to get the projection counts on $\eone$ and $\etwo$, 
we also need to explicitly record the mappings for the boundary nodes.

The two nodes $a_{\eone}$ and $a_{\etwo}$ may occur on either $\Pp_{h,d}$ or $\Pm_{h,d}$.
Six different configurations are possible, of which two are  shown in Figure~\ref{fig:confy} (b). 
In Configuration (A), the paths include one boundary each, whereas in the second configuration, 
the same path includes both the boundary nodes.
The other four configurations are symmetric: the boundary nodes may swap the paths in which they occur
and in Configuration (B) can also reverse the order in which they occur.
We discuss the two configurations shown in the figure; 
the other configurations are handled in a similar fashion.

Consider configuration (A). 
In order to record the mappings of the boundary node $a_{\eone}$, 
we introduce an additional field in the projection counts.
For a triple of data vertices $u$, $v$ and $x$, and a signature $\alpha$,
let $\cnth(u,v,x, \alpha|\Pp_{h,d})$ denote the number of high-starting matches
$\pi$ for $\Pp_{h,d}$ with $\pi(a_h)=u$, $\pi(a_d)=v$, $\pi(a_{\eone})=x$ and $\sig(\pi)=\alpha$.
These counts are computed in a manner similar to the base procedure shown in Figure \ref{fig:db-algo} (Procedure 1);
however, when the process encounters the boundary node $\eone$ (namely, the initialization step or $j = \eone$),
the mapped vertex ($v$ or $w$, respectively) is recorded in the additional field.
The analogous counts for $\Pm$ can derived in a similar manner.
The value of $\cnth(u,v,\alpha|\calC,\hi=h)$ is obtained by joining the two; 
see Procedure (2) in Figure~\ref{fig:db-algo}.
Configuration (B) is handled in a similar fashion, except that we need two additional fields to 
record the mappings for both the boundary nodes. Namely, we maintain counts having keys of the form
$(u,v,x,y)$ representing the mapping of the nodes $h,d,\etwo$ and $\eone$ to the
vertices $u,v,x$ and $y$. Procedure (2) is also adjusted accordingly.
Finally, we can get the projection table $\cnt(u,v,\alpha|\calC)$ 
via aggregation, as in Equation  \ref{eqn:aggr}.

\begin{figure}[t]
\begin{center}
\begin{boxedminipage}{\hsize}
\begin{small}
\begin{tabbing}
xx\=xx\=xx\=xx\=xxx\=xxx\=\kill
\textbf{Compute Projection Table for $\Pp_{h,d}$}\\
Let $B$ be the block annotating the edge $(a_h, a_{h\oplus 1})$\\
$\cnth(\cdot,\cdot,\cdot|\Pp_{h,h\oplus 1}) = \cnth(\cdot,\cdot,\cdot|B)$\\
For $j=h\oplus 1,h\oplus2, \ldots, d$\\
\> Execute $\NodeJoin(a_j$)\\
\> Execute $\EdgeJoin(a_j$)\\
Execute $\NodeJoin(a_d)$\\
\\
$\NodeJoin(a_j)$:\\
If $a_j$ is annotated by a block $B$\\
\> For each $(u,v,\alpha_1)$ with $\cnth(u,v,\alpha_1|\Pp_{h,j})\neq 0$\\
\> \> For each $(v,\alpha_2)$ with $\cnt(v, \alpha_2|B)\neq 0$ \\
\> \> \> If $(\alpha_1 \cap \alpha_2 = \{\chi(v)\}$ \\
\> \> \> \> $\alpha \leftarrow \alpha_1 \cup \alpha_2$\\
\> \> \> \> $\val_1 \leftarrow \cnth(u,v,\alpha_1|\Pp_{h,j})$; \quad $\val_2 \leftarrow \cnt(v, \alpha_2|B)$\\
\> \> \> \> Incr. $\cnth(u,v,\alpha|\Pp_{h,j})$ by $\val_1\times \val2$\\
\\
$\EdgeJoin(a_j)$\\
For each entry $\cnth(u,v,\alpha_1|\Pp_{h,j})\neq 0$\\
\> For each entry $\cnt(v, w, \alpha_2|B)\neq 0$ and $u \succ w$\\
\> \> If $(\alpha_1\cap \alpha_2 = \{\chi(v)\}$\\
\> \> \> $\alpha \leftarrow \alpha_1 \cup \alpha_2$\\
\> \> \> $\val_1 \leftarrow \cnth(u,v,\alpha_1|\Pp_{h,j})$; \quad $\val_2 \leftarrow \cnt(v,w,\alpha_2|B)$\\
\> \> \> Incr. $\cnth(u,w,\alpha|\Pp_{h,j\oplus 1})$ by $\val_1 \times \val_2$
\end{tabbing}
\end{small}
\end{boxedminipage}
\end{center}
\caption{$\DB$ Procedure for General Cycle Blocks}
\label{fig:db-gen-algo}
\end{figure}

%

\subsection{Solving General Blocks}
\label{sec:gen-cycle}
In this section, we present procedures for handling generic blocks. We first consider the case of 
cycle blocks with two boundary nodes.

Consider a generic cycle $\calC=(a_0, a_1, \ldots, a_{L-1})$ having two boundary nodes $a_{\eone}$ and $a_{\etwo}$,
whose nodes and edges may be annotated with other blocks (children of $\calC$ in the decomposition tree). 
All these blocks have at most two boundary nodes and these are found on $\calC$.
For such any block $B$, the subquery represented by $B$ has the same boundary nodes as that of $B$.
Thus, we can get the projection table for $\calC$ by joining the projection
tables of the subqueries represented by the above blocks, as described below.

As before, we consider each possible choice for the highest node $a_h$
and split the cycle into two paths $\Pp_{h,d}$ and $\Pm_{h,d}$.
The path segment $\Pp_{h,d}$ also represents a subquery 
(induced by the union of the nodes found in the path and the blocks annotating path). 
Thus, we can extend the notion of projection
tables for these segments as well.
The procedure for computing the projection table for $\Pp_{h,d}$ is similar
that the one discussed in previous section (Procedure {1} in Figure \ref{fig:db-algo}),
and works by extending one edge in each step.
However, two aspects need to be addressed. 
Firstly, in contrast to the prior procedure, 
the edge being extended may be annotated with a block or 
un-annoated (and correspond to an original edge found in
input query $\inpQ$).
In the former case, we perform a join operation with the edges of the data graph (as before),
whereas in the latter case the join operation involves the projection table of the block $B$.
For the sake of uniformity, it will be convenient to view the former edges 
as blocks as well, denoted $B_{G}$, and associate with them a projection table derived from the
graph edges, as follows. 
For each edge $(u,v)\in G$, set $\cnt(u,v,\alpha)$ as $1$, for $\alpha=\{\chi(u),\chi(v)\}$; 
all other entries of the table are set to a count of zero.
The second aspect is that the nodes of the cycles may also be annotated,
and these get included as part of the sequence of joins being performed.
The two aspects are addressed by procedures called $\NodeJoin$ and $\EdgeJoin$.
The pseudo-code is shown in Figure \ref{fig:db-gen-algo}.

The procedure starts with an initial table representing the first edge $(a_h, a_{h\oplus 1})$
and performs a sequence of join operation with the blocks annoatating the nodes 
and edges of the cycle. 
At this juncture, two intricacies must be highlited. Firstly, the endpoint $a_h$ and/or $a_d$ may be annotated
by a block $B$, which must be joined by either $\Pp_{h,d}$ or $\Pp_{h,d}$,
but not by both (to avoid double counting). For this purpose, we adopt the convention that 
$\Pp_{h,d}$ and $\Pm_{h,d}$ include only the block annotating 
$a_d$ and $a_h$ (if found), respetively. 
Secondly,  for a block with two boundary nodes $\eone$ and $\etwo$,
the projection table views one of them as the first boundary node
and the other as the second (corresponding to the two components
of the keys of the form $(u,v,\alpha)$).
Thus, the boundary nodes are ordered and the projection tables need
not be symmetric: taking $\etwo$ as the first boundary node
and $\eone$ as the second boundary node would produce a different boundary tables.
However, the boundary tables are transpose of each other ($\cnt(u,v,\alpha) = \cnt(v,u,\alpha)$).
Our algorithm maintains both the tables and uses the appropriate one 
as dictated by the nodes of the cycle.
The pseudo-code reflects the first aspect, but, for the sake of clarity, ignores the second.

The projection counts obtained by the above process are joined using a proecure
similar to Figure~\ref{fig:db-algo}, taking into account the configuration in which the boundary nodes occur.
These are aggregated over all possible choices of the high node $a_h$.

Cycles with a single boundary node are handled in a similar manner
by considering each possible choice for the highest node $a_h$
and splitting the cycle into two paths $\Pp_{h,d}$ and $\Pm_{h,d}$.
The setting is simpler with only two configurations possible on how the boundary nodes may 
appear on the paths: the (single) boundary node may appear in $\Pp$ or $\Pm$.
Thus, the prior procedures can be applied here as well.

The case of leaf blocks are also handled via join operations.
Any leaf block $(a,b)$ is processed by joining the projection table for 
the blocks annotating the nodes $a$, the edge $(a,b)$ and the node $b$ (if found).

At  the end of the traversal process, the root block is solved,
which is either a cycle or a singleton node. 
In the former case, the block is treated as a cycle without boundary nodes.
Instead of computing its projection table, we simply count the number of colorful matches,
via a procedure similar to that of two-boundary cycles.
In the latter case, we consider the projection table of the block annotating the singleton node 
and output the sum of counts across all entries of the table.
The process yields the number of colorful matches of the input query $\inpQ$.
\section{Finding Good Decomposition Trees}
\label{sec:planner}
In each step of the decomposition process, multiple blocks may be available for contraction.
Each sequence of choices leads to a unique decomposition tree,
and hence, multiple trees are possible for a given query.
For example, the query ${\tt brain1}$ (Figure~\ref{query-pic}) admits two decomposition trees:
(i) contract the $4$-cycle first and then the $6$-cycle, and (ii) vice versa.
We conducted an experimental study involving a number of real-world data graphs and queries.
For each query, we enumerated all the possible decomposition trees and evaluated the execution time on each graph.
We observed a maximum difference of $13$x in the execution times 
of two decomposition trees for the same graph-query combination.
However, we noted that in most cases the optimal tree is independent of the data graph
and is mainly determined by the structure of the query.
These observations show that we need a procedure for selecting a good tree,
but in this process, we need not analyze the large data graph;
rather, it suffices to focus on the structural properties of the small query graph. 

Our study also showed that the following factors, in the decreasing order of importance, determine the execution time: 
(i) length of the longest cycle block; 
(ii) number of boundary nodes;
(iii) number of node/edge annotations.
Armed with the above observations, we designed a simple heuristic procedure.
Enumerate all possible trees for the given query and pick the best using the above factors for comparison.
In our experimental setting, barring a few exceptions, 
the heuristic picked the optimal tree in majority of the cases and a near-optimal tree for the rest.
Since the queries are of small size (about 10 nodes), even a sequential implementation of the 
heuristic takes insignificant amount of running time.

\section{Distributed Implementation}
\label{distr_impl}
In this section, we present a brief sketch of the distributed implementation of the two algorithms, 
highlighting their main aspects.
The distributed implementation consists of three layers. The first
layer, called the {\em planner}, finds a good decomposition tree for
the given query a fast sequential implementation the heuristic
discussed in Section \ref{sec:planner}. The second layer, called the
{\em plan solver}, takes the data graph and the decomposition tree
and implements the $\PSE$ and $\DBE$ algorithms presented in
Section \ref{sec:solving-blocks}. It accomplishes the above task by using efficient
join routines supported by the third layer, called {\em engine}.
The engine has three functionalities. The first is to store
the data graph in a distributed manner. This is achieved
via a $1D$ decomposition, wherein the vertices are equally
distributed among the processors using block distribution,
and each vertex is owned by some processor. The second is
to maintain projection tables. These tables are of two
types: unary projection tables having single-vertex keys
of the form $(u, \alpha)$ associated with blocks having single boundary nodes; 
binary projection tables having two-vertex keys of the form
$(u, v, \alpha)$. The binary tables also have variants involving
additional fields for storing the mappings for the boundary
vertices. The engine provides a convenient abstraction to the
plan solver for all these types of tables. All the
tables are maintained as distributed hash tables which use
open addressing to resolve collisions. Every entry $(u, v, \alpha)$ is
stored on the processor owning $v$; the degree of $v$ is packed as
part of the entry for enforcing the degree constraint in the
join operations (of the form $u \succ w$ in Procedure 1 of
Figure~\ref{fig:db-algo}). Signatures are maintained as bitmaps.
The third functionality is to support two types of join
operations on the projection tables. The first type of join
is used for extending a path segment an edge; this involves
a join with either the graph edges or the projection table
of the block annotating the edge. In the former case, the
extension of an entry with a key $(u, v, \alpha)$ with an edge
$(v,w)$ will be performed at the owner of $v$. The result is an
entry with a key $(u, w, \alpha')$; this entry is communicated to
the owner of $w$, where it gets stored. The latter case involves
join of two entries with keys $(u, v, \alpha_1)$ and $(v, w, \alpha_2)$. Since
the first entry is stored at the owner of $v$ and the second,
at the owner of $w$, a communication is
performed to bring the two entries to a common processor.
The second type of join is used for merging the projection
tables of two path segments (for example, Procedure 2
in Figure~\ref{fig:db-algo}) and it is implemented in a similar way. 
The two operations are implemented using a standard sort-merge join
procedure with signature compatibility checks performed via
fast bitwise operations.

\begin{table}[t]
\centering
\caption{Real Data Graphs}
\begin{tabular}{l|l|l|l|l|l}
{\bf Graph}  & {\bf Domain}   & {\bf Nodes} & {\bf Edges} & {\bf Avg}  & {\bf Max} \\
& & & & {\bf Deg}  & {\bf Deg} \\ \hline
{\tt brightkite}  & Geo loc.   & 58K   & 214K  & 4 & 1135 \\ 
{\tt condMat} & Collab.  & 23K   & 93K   & 4 & 281 \\ 
{\tt astroph}    & Collab.  & 18K   & 198K & 11 & 504   \\ 
{\tt enron}      & Commn.  & 36K   & 180K  & 5 & 1385\\ 
{\tt hepph}      & Citation       & 34K   & 421K  & 12 & 848  \\ 
{\tt slashdot}   & Soc. net. & 82K   & 900K  & 11 & 2554 \\
{\tt epinions}   & Soc. net. & 131K  & 841K  & 6 & 3558 \\ 
{\tt orkut}   & Soc. net.   & 524K   & 1.3M  & 3 & 1634 \\ 
{\tt roadNetCA} & Road net.   & 2M    & 2.7M & 1.3 & 14\\ 
{\tt brain} & Biology   & 400K    & 1.1M & 3 & 286
\end{tabular}
\label{graph-tab}
\end{table}

\begin{figure}
\centering
\includegraphics[scale=1]{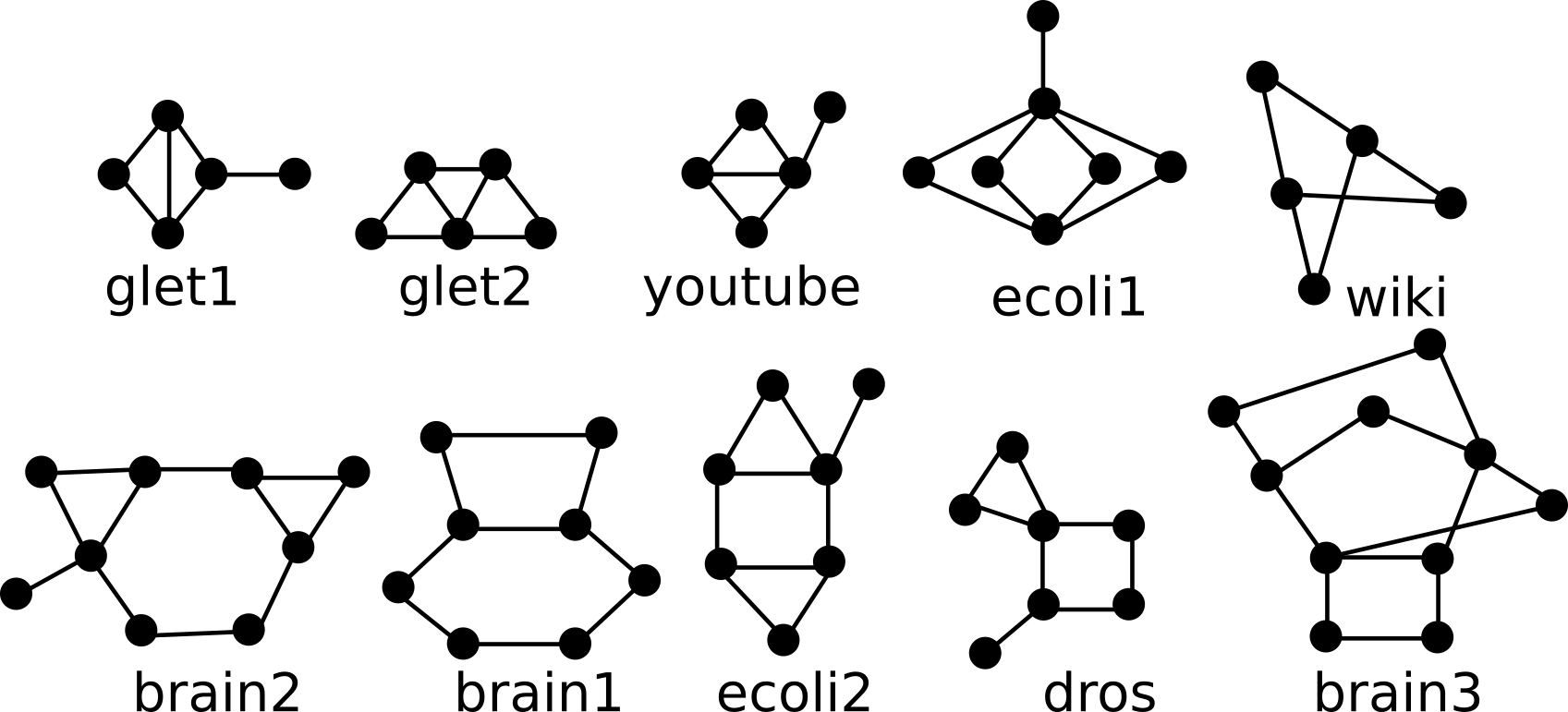}
\caption{Real world queries used in our study.}
\label{query-pic}
\end{figure}

\section{Experimental Study}\label{sec:expt}
\label{sec:expt_setup}
We present an extensive experimental evaluation of the algorithms presented in the paper. 
Our experiments include a comparison of the algorithms on execution time, strong and weak scaling studies for our algorithm, 
and studies to evaluate the quality of our query plan generation heuristic and the efficacy of color coding for treewidth two queries. 

\subsection{Experimental Setup}
{\bf System. }
The experiments were conducted on an IBM Blue Gene/Q system \cite{bgq1}. Each BG/Q node has $16$ cores and 
$16$ GB memory; multiple nodes are connected using a $5$D torus interconnect.  
Our implementation is based on MPI2 with {\tt gcc 4.4.6} with the number of ranks varying from $32$ to $512$. 
Each MPI rank was mapped to a single core. The number of MPI ranks mapped to a node was adjusted based on 
the memory requirements of individual experiments.

{\bf Graphs.}
The experiments involved nine real world graphs obtained from
the SNAP dataset collection and the human brain network from 
the Open Connectome Project (\url{http://snap.stanford.edu}, \url{http://www.openconnectomeproject.org/}).
Our benchmark includes representative graphs from different domains in SNAP.
The graphs and their characteristics are presented in Table~\ref{graph-tab}.  
We also used synthetic R-MAT graphs~\cite{RMAT}, for the purpose of studying the weak scaling behavior of our algorithms.

{\bf Queries. }
Our query benchmark consists of the ten real world queries shown in Figure~\ref{query-pic}.
The queries were derived from prior network analysis work spanning diverse domains:
{\tt dros, ecoli1, ecoli2, brain1, brain2, brain3} - biological networks \cite{Middendorf01032005, Iakovidou2013204};
{\tt glet1, glet2} - graphlets \cite{6413912}; 
{\tt wiki} - collaboration networks \cite{wu2012cwaunmcar};
{\tt youtube} - spam networks \cite{ytube}.

{\bf Algorithms. }
We study two algorithms: $\PSE$, which serves as the baseline, and our degree based $\DBE$ algorithm.
Recall that $\PSE$ is equivalent to the dyamic programming based algorithm of Alon et al. \cite{AYZ95}.

\begin{figure}
\centering
\includegraphics[scale=1]{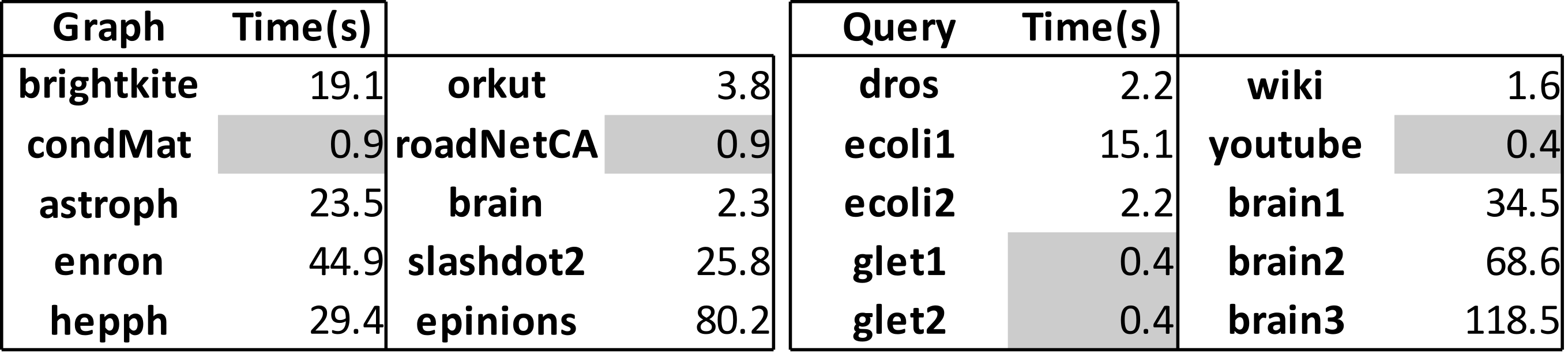}
\caption{Average execution time (seconds).}
\label{avg-run-times}
\end{figure}

\subsection{Graph-Query Characteristics}
\label{sec:summary_stats}
The characteristics of the input graph and query strongly influence the running time of query counting 
algorithms. To obtain an overall characterization of the phenomenon, we measured the execution time of 
the $\DBE$ algorithm on each of the 100 real graph and query combinations using 512 MPI ranks.
Figure~\ref{avg-run-times} shows the average running time for each graph
across the ten queries and the average running time of each query across the ten graphs.
The wide variations in execution time across graphs and queries is indicative of their relative difficulty in practice. 
For example, although {\sf roadNetCA} is a larger graph than {\sf epinions}, the average running time of the former is 
smaller than the latter by an order of magnitude. We can understand this behaviour by studying the skew in underlying degree distribution. 
In general, counting colorful occurrences of a query on a graph with high skew (indicated by high maximum degree in Table~\ref{graph-tab}) 
tends to be computationally expensive. Similarly, the queries also exhibit large variations in running time, ranging from sub-second
for {\sf youtube, glet1} and  {\sf glet2} to more than a minute for {\sf brain2} and {\sf brain3}. These variations can be accounted for 
by studying the differences in the size and the sub-structures of the queries. We observed that queries with longer cycles are more challenging. 
As an extreme case, a $12$-vertex complete binary tree query requires 2 seconds on average, 
in contrast to the $10$-vertex {\sf brain3} query which 
requires nearly 2 minutes on average, exemplifies our observation.
 
\subsection{Performance Comparison of $\PSE$ and $\DBE$ Algorithms}
We study the performance of the $\PSE$ and $\DBE$ algorithms on 
100 graph-query combinations obtained by selecting a graph from Table~\ref{graph-tab} and 
a query from Figure~\ref{query-pic}.
For our $\DBE$ algorithm, we used plans supplied by the heuristic described in Section~\ref{sec:planner}. 
In contrast, for the $\PSE$ algorithm, 
we enumerated all the possible plans and obtained the optimal plan. 
Thus, we compare our algorithm to the best possible scenario for the baseline algorithm.

We compute the improvement factor $(IF)$ of $\DBE$ over $\PSE$ as the ratio of the execution time of $\PSE$ to $\DBE$.  
Figure~\ref{plot:compare} shows $IF$ at $32$ and $512$ ranks.
The combinations where $\DBE$ outperforms $\PSE$ $(IF > 1)$ are highlighted in green. 
The blank entries represent cases where $\PSE$ (or $\DBE$) did not complete execution, due to lack of available memory. 
At $32$ ranks, we can see that $\DBE$ outperforms $\PSE$ on 84\% of the graph-query combinations 
with $IF$ being as high as $9.1$x (average $2.4$x). At $512$ ranks, 
$\DBE$ outperforms the baseline on 89$\%$ of the cases, with $IF$ becoming as high as $28.7$x (average $5.0$x). 

\begin{figure}
 \centering
 \begin{tabular}{c}
 \includegraphics[scale=1.0]{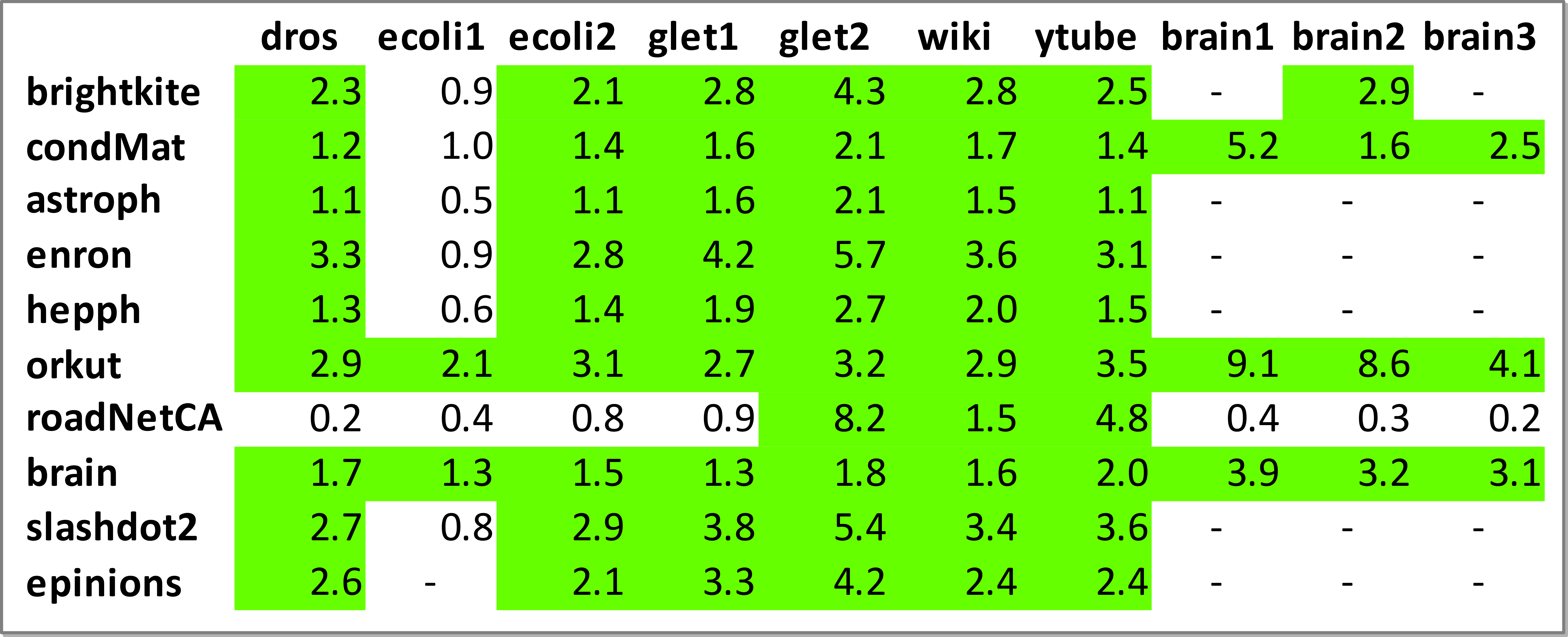}
 \\
 \\
 (a) $32$ Ranks
 \\
 \\
 \includegraphics[scale=1.0]{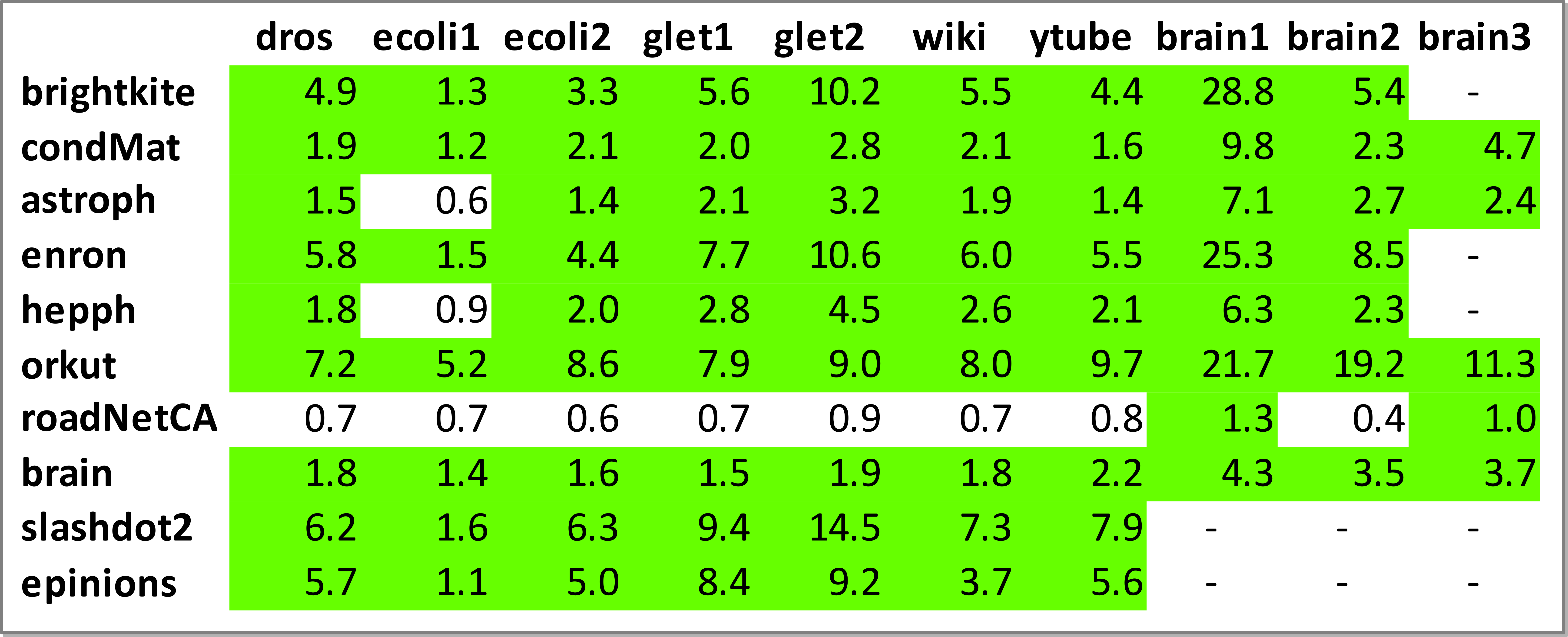}
 \\
 \\
 (b) $512$ Ranks
 \end{tabular}
 \caption{Improvement factor of the $\DBE$ algorithm over the $\PSE$ algorithm.}
 \label{plot:compare}
 \end{figure}

We can see that the relative performance of the two algorithms is dependent on the graph-query pair.
For instance, the average $IF$ on {\tt enron} and {\tt condmat} graphs are $8.4$ and $3.1$ on $512$ ranks, respectively,
correlating well with their skew in the degree distribution (see Table \ref{graph-tab}).
Similarly, the improvement factors is higher on complex queries such as {\tt brain1} where the average improvement is $13.1$x,
compared to {\tt youtube} where the average improvement is only $4.1$x.
The phenomenon becomes extreme in the case of road networks that have very low skew
and exhibit sub-second average running time across queries. 

Our $\DBE$ algorithm scales better than $\PSE$, as demonstrated by the increase in $IF$ at higher ranks.
For different graph-query combinations, we computed the ratio of $IF$ at $512$ ranks to that of $32$ ranks
and found that $IF$ increases by a factor of up to $4.7$x (average $1.7$x).
To understand this trend further, we compute the load (number of projection function operations) for both algorithms 
for processing different queries on the {\tt enron} graph at 512 ranks. 
For different queries, Figure \ref{max-load-pic} shows the execution time and the average and maximum load.
We can see that $\DBE$ has lesser average load than $\PSE$, since $\DBE$ avoids wasteful computations.
Furthermore, the improvement obtained by $\DBE$ over $\PSE$ on execution time correlates 
well with improvement obtained on the maximum load.
For example, on {\tt ecoli1} query, 
even though $\PSE$ outperforms $\DBE$ at $32$ ranks, 
the perforamance is reversed at $512$ ranks (see Fig \ref{plot:compare}),
because of superior load balancing characteristic of $\DBE$.

\begin{figure}
\centering
\begin{tabular}{c}
\begin{tabular}{ccc}
\includegraphics[scale=0.33]{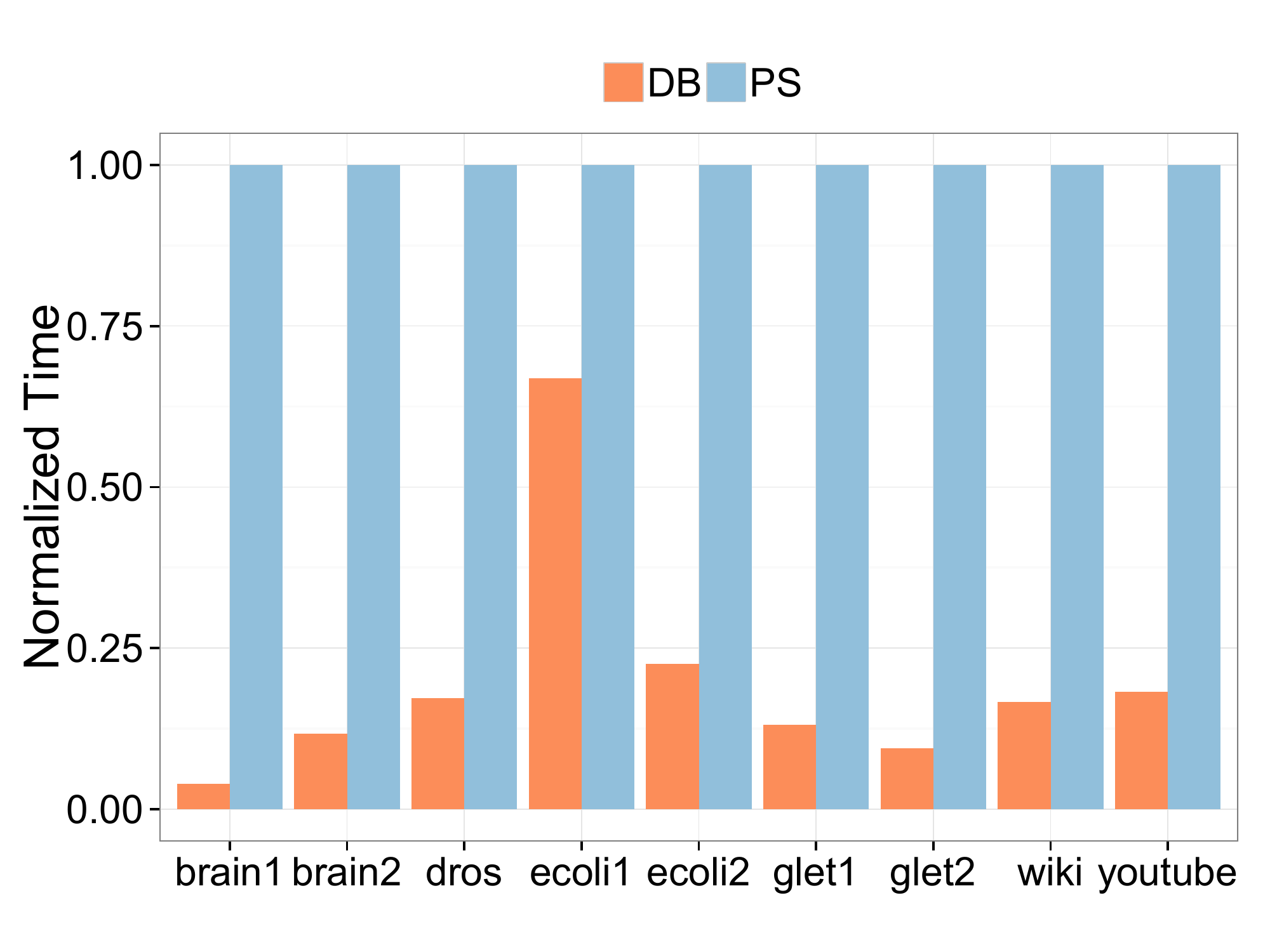} 
& \quad&
\includegraphics[scale=0.33]{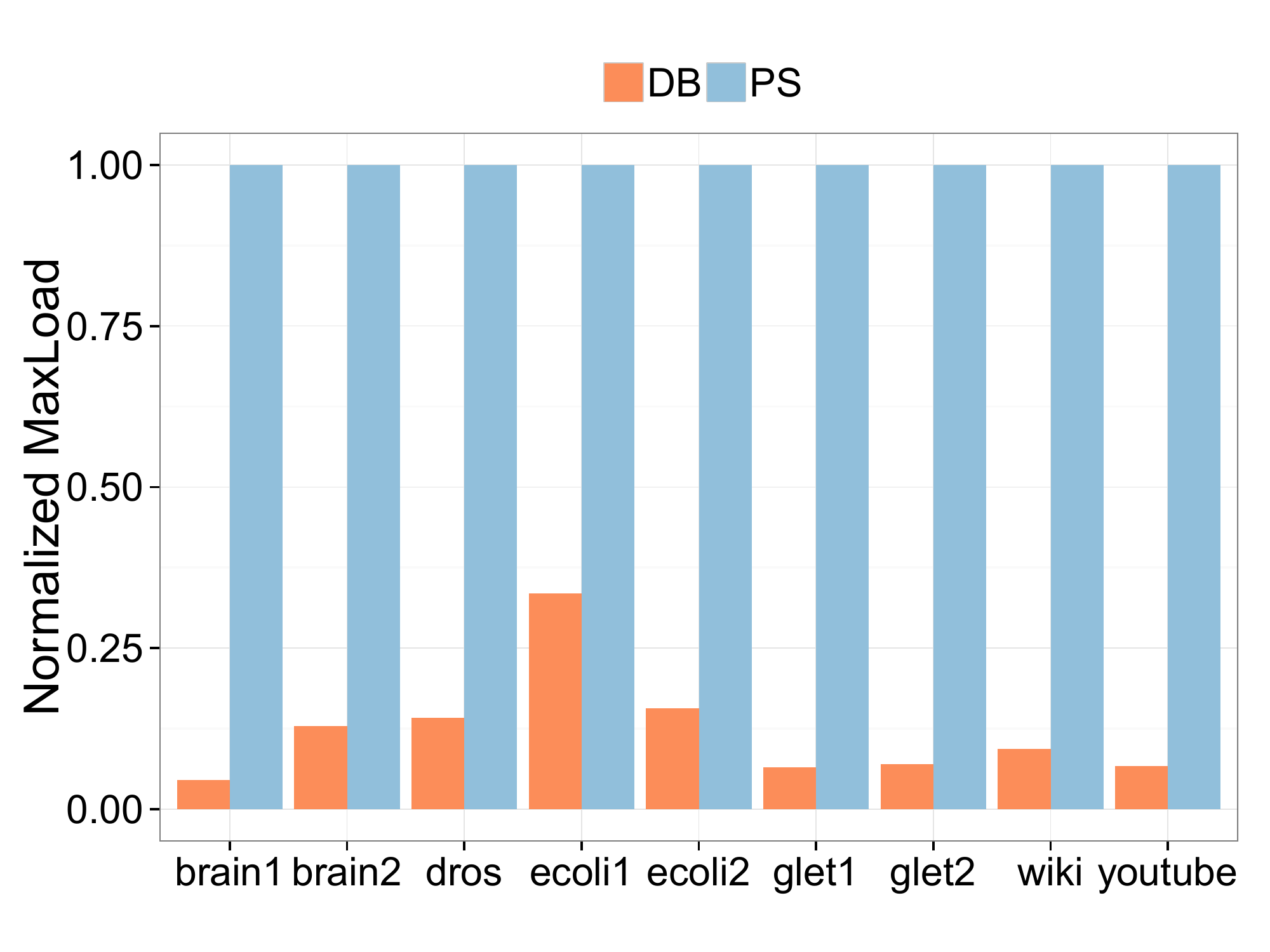}
\\ \\
(a) Time
& \quad\quad &
(b) Max. Load
\end{tabular}
\\
\\
\includegraphics[scale=0.33]{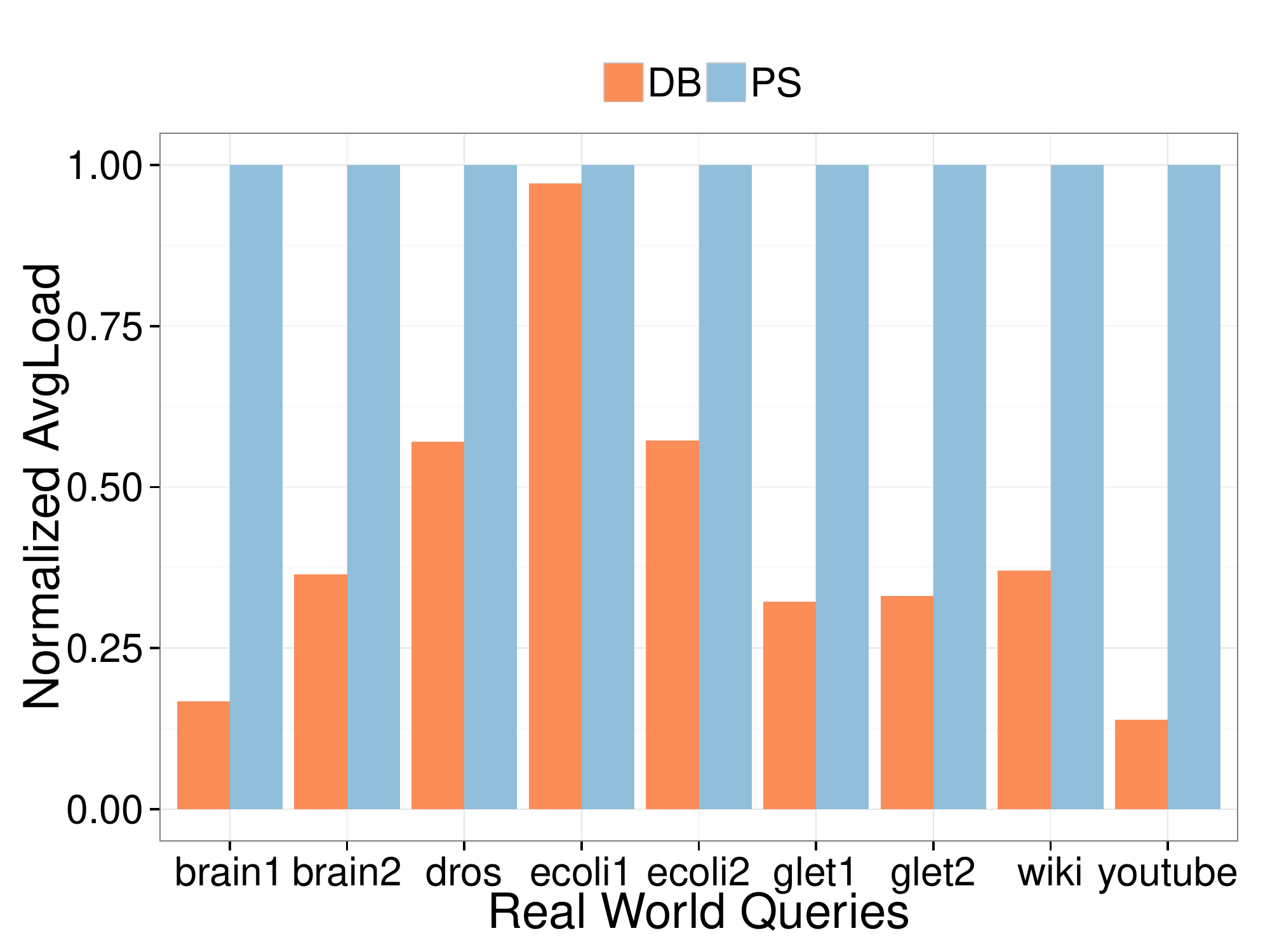}  
\\\\
(c) Avg. Load
\end{tabular}
\caption{Normalized execution time, average load and maximum load on {\tt enron} graph.}
\label{max-load-pic}
\end{figure}

\begin{figure}
\centering
\includegraphics[scale=1]{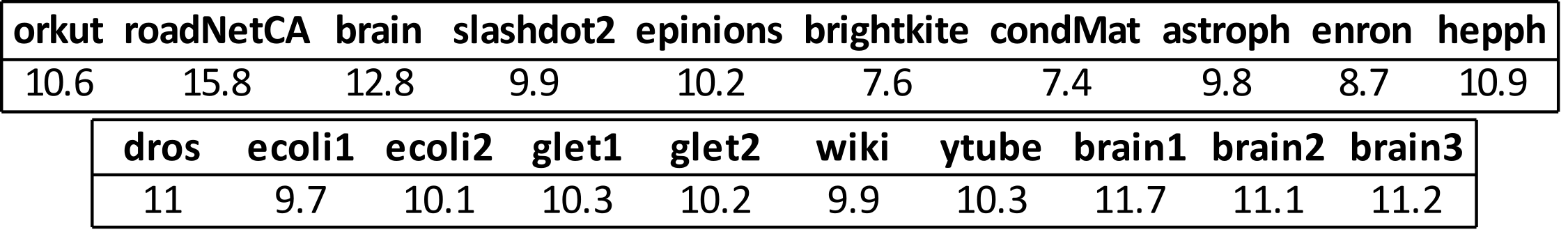}
\caption{Avg. speedup of $\DBE$ at 512 ranks compared to 32 ranks.}
\label{speedup-table}
\end{figure}

\subsection{Scalability Characteristics of $\DBE$ Algorithm}
We studied the scaling of $\DBE$ across the 100 graph-query combinations. 
For each combination, we computed the ratio of the execution time at $512$ ranks to that of $32$ ranks.
Figure~\ref{speedup-table} summarizes the above information by providing the averge of the above speedup for each query across graphs
and the same for each graph across queries. As against an ideal speedup of $16$x, we see that the algorithms
obtains speedups in the range of $7.4$x to $15.8$x.

We studied the strong scaling behavior of our algorithm, using {\tt enron} as a representative graph.
Taking $32$ ranks as the baseline, Figure \ref{scaling-pic} shows the speedup up to $512$ ranks for different queries.
The algorithm scales well across queries, with an average speedup of $8.2$x and maximum speedup of $9.9$x at $512$ ranks
(as against an ideal speedup of $16$x).

\begin{figure}
\centering
\begin{tabular}{lll}
	\includegraphics[scale=0.35]{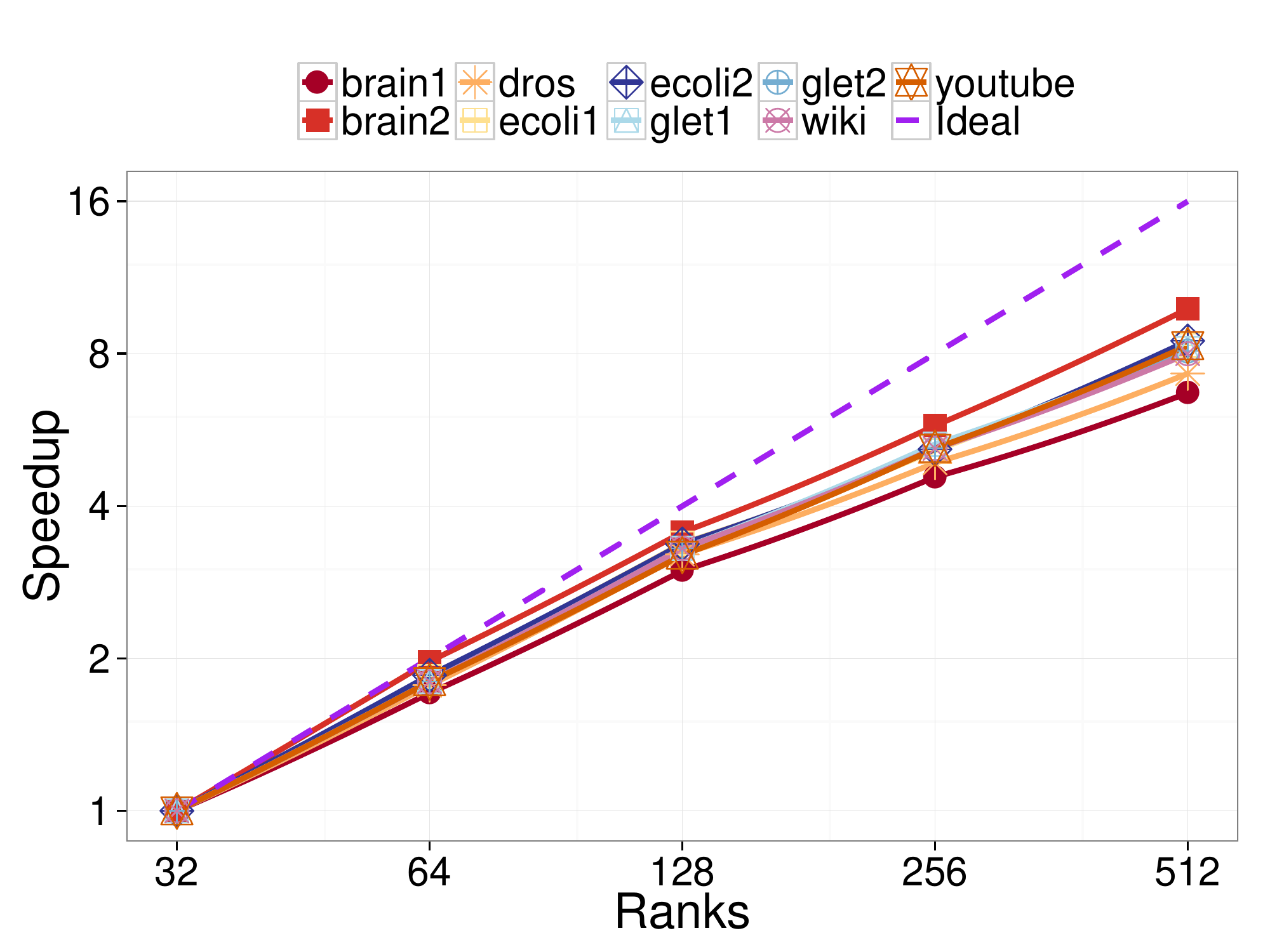} & \quad 
	\includegraphics[scale=0.35]{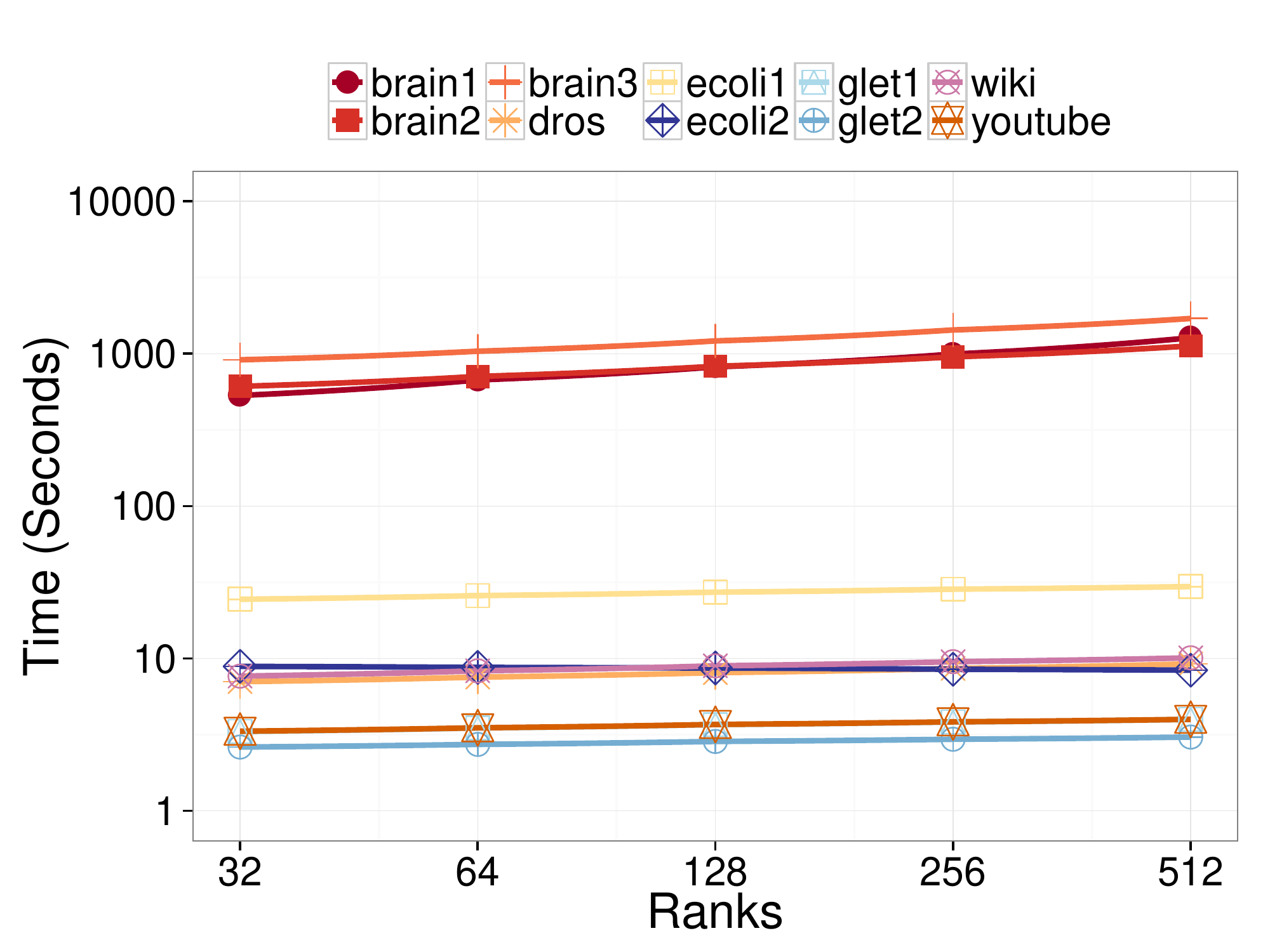} 
 \end{tabular}
\caption{Strong and weak scaling} 
\label{scaling-pic}
\end{figure}

To study weak scaling, we use R-MAT synthetic graphs with parameters 
$A=0.5$, $B=0.1$, $C=0.1$ and $D=0.3$ and edge factor $16$,
suggested in a Graph 500 benchmark
specification (\url{http://www.cc.gatech.edu/~jriedy/tmp/graph500/}).
The number of vertices was fixed at $1K$ per rank and the number of ranks was varied from 32 to 512. 
We report the execution times each query-rank combination in Figure~\ref{scaling-pic}. 
We see excellent weak scaling behavior with the execution times at $512$ ranks remaining close to that of the baseline 32 ranks.

\subsection{Evaluation of Plan Generation Heuristic}
We studied the quality of our plan generation heuristic for the $\DBE$ algorithm at $512$ ranks.
For each graph-query combination, we determined the optimal plan via an exhaustive enumeration.
We compared the execution time of the heuristic plan to the optimal plan and measured the percentage difference.
These results are reported in Figure \ref{plan-accuracy-pic}.
We can see that in $90$\% of the case, the heuristic generated the optimal plan,
whereas in the remaninig cases, the difference was at most $15$\%.

\subsection{Precision of Color Coding}
We evaluated the precision of color coding on our benchmark by
performing independent trials and computing the empirical variance of
the sample (see Section~\ref{sec:prelims}).  Specifically, for a given
graph-query combination we performed a sequence of trials, where in
each trial the colorful count $n^{colorful}(G, \inpQ, \chi)$ was
computed for a fresh random coloring.  We performed 10 random trials
for each of the 100 graph-query combinations in our test set and
evaluated the empirical mean and variance of the number of colorful
matches. For each graph-query combination, we computed the 
coefficient of variation, which is the ratio of the empirical variance
to the mean. The results are shown in Figure \ref{colorcodepic}. 
A value close to 0 indicates the convergence
of our estimate to the true mean $n(G, \inpQ)$. We observed that with
only three trials, 82\% of the graph-query combinations had coefficient of variation
at most $0.1$; when the number of trials was increased to 10, it increases to 91\%.
Hence, using 512 ranks,
for a majority of the input graph-query combinations in our benchmark,
we require less than a minute to count the actual number of matches of
the query, with $\approx 10\%$ accuracy. We conclude that our $\DBE$
algorithm enables fast approximate counting of treewidth 2 queries for data graphs spanning various real domains.

\begin{figure}
\centering
\includegraphics[scale=0.8]{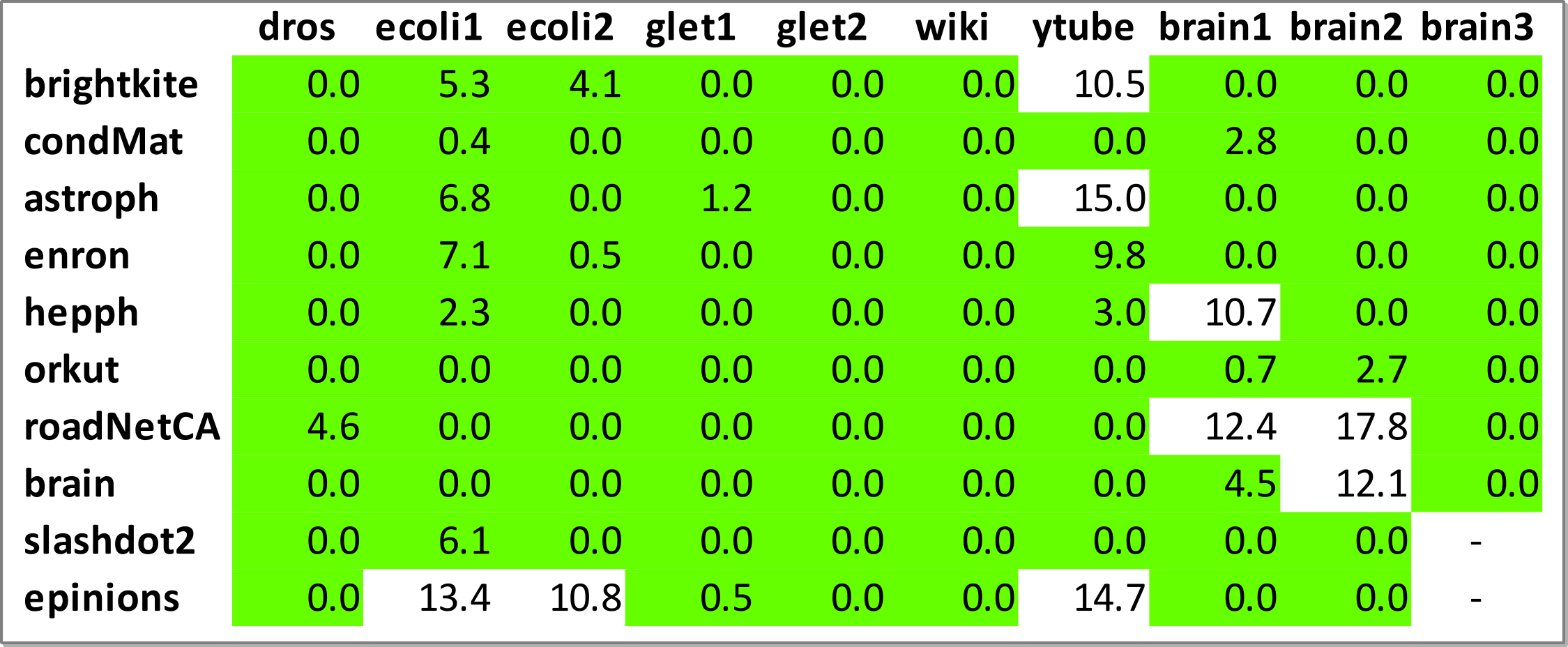}
\caption{Error \% of the execution time of the plan proposed by the plan heuristic with reference to the optimal plan for each graph-query combination.}
\label{plan-accuracy-pic}
\end{figure}

\begin{figure}
\centering
\includegraphics[scale=0.8]{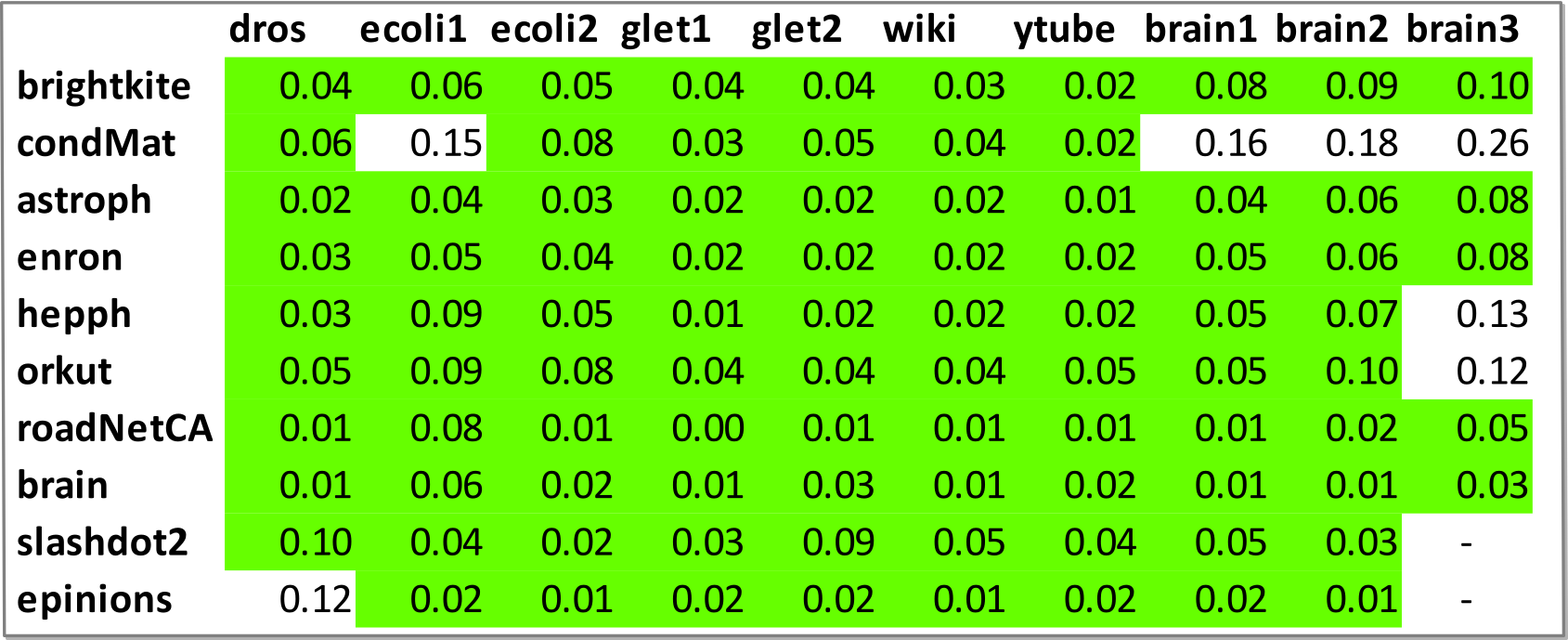}
\caption{Coefficient of variation with 50 trials of color coding for each graph-query combination.}
\label{colorcodepic}
\end{figure}

\newcommand{\expect}{{\bf \mbox{\bf E}}}
\newcommand{\prob}{{\bf \mbox{\bf Pr}}}
\newcommand{\E}{{\mathcal{E}}}
\newcommand{\calU}{{\mathcal{U}}}
\newcommand{\calW}{{\mathcal{W}}}
\renewcommand{\deg}{{\mathsf{deg}}}
\newcommand{\hdeg}{{\mathsf{\overline{deg}}}}
\renewcommand{\d}{\mathbf{d}}
\section{Cycle queries on random power law graphs}\label{sec:chung-lu}
In this section we concentrate on cycle queries of constant size and
analyze the expected runtime of a variant of the {\PS} and {\DB}
procedures on a certain class of random data graphs.
We prove a lower bound on the expected runtime of the {\PS} procedure and an
upper bound on the expected runtime of the {\DB} procedure. Both bounds are functions of the (expected) degree sequence of the graph. We show that our upper bound on the runtime of the {\DB} procedure is never worse (up to constant factors) than the lower bound on the {\PS} procedure. Moreover, if the random graphs satisfy a natural power law property
then we prove that the expected runtime of the {\DB} procedure is
polynomially better.
Recall that the most complicated blocks in our general decomposition
of the query graph are (annotated) cycles. Thus, we postulate
that the better performance of
our variant of the {\DB} procedure on cycle queries explains the better performance
of the {\DB} procedure studied in the main body of the paper on queries of treewidth 2.

The class of random graphs considered is a certain variant of the
Chung-Lu graphs~\cite{Chung10122002}, a popular model for random
graphs that captures several properties of real world social
networks, as defined precisely below.

\subsection{The procedures analyzed}\label{sec:procedures}

Consider a cycle query $\calC$ of constant size $k$.
Since the query graph is a cycle, the decomposition tree consists
of only a root node which represents this cycle with a single boundary node.

Recall that procedure {\PS} computes $\cnt(u,\{1,\dots,k\}|\calC)$, for each
node $u$. 
For this it first computes $\cnt(u,v,\alpha_1|\Pp)$ and $\cnt(u,v,\alpha_2|\Pm)$ for
all of nodes $v$ and all signatures.
Since $k$ is constant the number of signatures is also constant and
thus computing $\cnt(u,\{1,\dots,k\}|\calC)$ given 
$\cnt(u,\cdot,\cdot|\Pp)$ and $\cnt(u,\cdot,\cdot|\Pm)$ can be done in
linear time.
In the {\PS} procedure the computation of $\cnt(u,\cdot,\cdot|\Pp)$ and
$\cnt(u,\cdot,\cdot|\Pm)$ is done by iteratively recomputing
$\cnt(u,v,\cdot|\Pp)$ for all $u$ and $v$. This recomputation can be
viewed as a optimized version of enumerating all paths starting at $u$, where instead of storing paths explicitly, we only store the two endpoints, the signature, and a count.
Our simplified variant of the {\PS} procedure, which is more amenable to analysis, 
computes $\cnt(u,\cdot,\cdot|\Pp)$ and
$\cnt(u,\cdot,\cdot|\Pm)$ by enumerating all possible
paths, instead of performing this optimized enumeration.
Thus, its complexity is linear in the number of possible 
paths of lengths $1,\ldots,\ceil{\half k}$. We will refer to this version of {\PS} procedure as the {\PS} procedure throughout this section to simplify notation.

To increase the efficiency of our {\PS} procedure when applied to cycle
queries we can
break symmetry using the id of the nodes and count only colorful matches
where node $u$ has the largest id among all
data nodes in the image of $\pi$. Consequently, $\cnt(u,\cdot,\cdot|\Pp)$ and
$\cnt(u,\cdot,\cdot|\Pm)$ need only to count paths with the property
that node $u$ has the largest id among nodes on the path.

For an integer $q \ge 3$ let
\begin{equation}\label{eq:y-def}
Y(q)=\left|\left\{(u_1, \ldots, u_q)\text{~is a simple path and~}\text{id}(u_1)>\text{id}(u_j), j\in [2..q]\right\}\right|.
\end{equation}
It follows that the expected complexity of procedure {\PS} is linear in
$\sum_{q=1}^{\ceil{k/2}} \expect[Y(q)]$. Later, we derive lower bounds on $\expect[Y(q)]$ for any constant $q$, and our bounds are monotone increasing in $q$. Thus, the dominant term in the complexity of procedure {\PS} is provided by our bound on $Y(\ceil{\half k})$.

Procedure {\DB} also computes $\cnt(u,\{1,\dots,k\}|\calC)$, but does it by
computing high-starting colorful matches.
Again, to increase the efficiency of procedure {\DB} when applied to cycle
queries we can break symmetry and only compute
$\cnt(u,\{1,\dots,k\}|\calC, \hi=u)$, namely, colorful matches where node $u$ is
the highest
in the degree based ordering among all data nodes in the image of $\pi$.
It follows that for this we need only to compute
$\cnth(u,\cdot,\cdot|\Pp_{h,d})$ and $\cnth(u,\cdot,\cdot|\Pm_{h,d})$,
namely,  the number of high-starting 
colorful matches for the paths $\Pp_{h,d}$ and $\Pm_{h,d}$,
wherein $\pi(a_h)=u$, $\pi(a_d)=v$ and $\sig(\pi)=\alpha$, for all
possible $v$ and $\alpha$. Note that similarly to the {\PS} procedure, the recomputation in the {\DB} procedure considered in the main body of the paper can be viewed as a optimized enumeration, where instead of storing paths explicitly, we only store the two endpoints, the signature, and a count.
Our simplified variant of the {\DB} procedure, 
which is more amenable to analysis, 
computes the counts by enumerating all these paths,
instead of performing this optimized enumeration.

For an integer $q\geq 3$ let 
\begin{equation}\label{eq:x-def}
X(q)=\left|\left\{(u_1,\ldots,u_q)\text{~is a simple path and~}u_1\succ u_j, j\in [2..q]\right\}\right|.
\end{equation}
It follows that the expected complexity of procedure {\DB} is linear in
$\sum_{q=1}^{\ceil{k/2}} \expect[X(q)]$. Later, we will derive bounds on $\expect[X(q)]$, and since our bounds will be monotone increasing in $q$, the dominant term in complexity of procedure {\DB} is provided by our bound on 
$\expect[X(\ceil{\half k})]$.
\eat{
We start with
\begin{definition}[Degree based ordering]
Let $\prec$ denote the degree ordering on nodes of $V$, i.e. we write $u\prec v$ if the degree of $u$ is smaller than the degree of $v$, with ties broken according to an arbitrary ordering.
\end{definition}
}

\subsection{The random data graphs}
We analyze our algorithm on Chung-Lu graphs whose expected degree
sequence has some additional property.

\paragraph{Chung-Lu distribution} The Chung-Lu distribution on random graphs is defined as follows. First choose a degree sequence $\d=(d_1,\ldots, d_n)$, where $V=[n]$. We assume that $d_u\geq 1$ for all $u\in V$. Let $m=\frac1{2}\sum_{u\in V} d_u$. We assume that $m\geq n$ and $\max_{u\in V} d_u\leq \sqrt{n}$. To generate the graph $G=(V, E)$ we include each edge $(u, v)$ independently with probability $d_u d_v/(2m)$.  Note that the expected degree of every node $u\in V$ is 
$\sum_{v\in [n]} d_u d_v/(2m)=d_u$ as required.
We write $\deg(u)$ to denote the actual degree of $u$ in $G$.

We require the degree sequence $\d=(d_1,\ldots, d_n)$ to be 
{\em $\lambda$-balanced}. 

\paragraph{Balanced degree sequence}
A degree sequence  $\d=(d_1,\ldots, d_n)$ is {\em $\lambda$-balanced} 
for some $\lambda\in (0, 1)$ if for any integers
$a, b\geq 1$,
$\sum_{u} d_u^{a+b}\leq \lambda \cdot (\sum_{u} d_u^a)(\sum_{u} d_u^b)$.
Intuitively, a degree sequence is balanced if it is not too concentrated on the high degree nodes.

For some of the claims we require the sequence to satisfy the stronger
truncated power law.
\paragraph{Truncated power law distribution}
A degree sequence $\d=(d_1,\ldots, d_n)$ satisfies the truncated power
law for a constant $\alpha\in (1, 2)$ if for each $0 \leq j \leq \half \log n$, 
the number of nodes with degree between $2^j$ and $2^{j+1}$ is $\Theta(n/2^{\alpha j})$. 
We show later that if a sequence satisfies the power law for $\alpha$
then it is  $\lambda$-balanced for $\lambda=O(n^{\alpha/2-1})$.

\subsection{Main theorem}

\begin{theorem}\label{thm:main}
Let $G$ be sampled according to the Chung-Lu distribution on $n$ vertices with an
$n^{-\delta}$-balanced degree sequence $\d$, for some constant
$\delta>0$, and let $q\geq 3$ be a constant. Then,
\begin{description}
\item[(1)] The expected number of paths $(u_1,\ldots,u_q)$ of length $q$ 
    for which $\text{id}(u_i)$ is with the highest id is lower bounded by
$$
\expect[Y(q)]\geq (1-o(1))\cdot \frac1{q}  (2m)^{-q+3}\left(\sum_{u} d_u^2\right)^{q-2}.
$$
\item[(2)] The expected number of high-starting paths of length $q$ is upper bounded by
$$
\expect[X(q)]\leq C (2m)^{-q+2} \left(\sum_{u\in V} d_u^{2-1/(q-1)}\right)^{q-1}
$$
for some constant $C>1$.
\item[(3)] Based on the above inequalities

\begin{enumerate}
\item $\expect[X(q)]=O(\expect[Y(q)])$.
\item If the degree sequence satisfies the truncated power law
   with parameter $\alpha$,  for any constant $\alpha\in (1, 2)$,
   then $\expect[X(q)]$ is polynomially smaller than $\expect[Y(q)]$.
\end{enumerate}
\end{description}
\end{theorem}
\begin{proof}
{\bf (1)} follows by Lemma~\ref{lm:y-bound}. {\bf (2)} follows by Lemma~\ref{lm:x-bound}. {\bf (3)} follows by putting together Lemma~\ref{lm:xleqy} and Lemma~\ref{lm:xy-powerlaw}.
\end{proof}

\begin{remark}

Note that if $\sum_{u} d_u^2 \geq \sum_{u} d_u=2m$ and 
$\sum_{u} d_u^{2-1/(q-1)} \geq \sum_u d_u =2m$, 
since $d_u\geq 1$ by our assumption on the degree sequence.   Thus  
both $\expect[Y(q)]$ and $\expect[X(q)]$ are monotone in $q$. 
It follows that for any constant $k$,
the dominant terms in the complexity of our {\PS} and {\DB} procedures for query cycles of
length $k$ is indeed determined by 
$\expect[Y(\ceil {\half k}\rceil)]$ and 
$\expect[X(\ceil {\half k}\rceil)]$, respectively,
as claimed in section~\ref{sec:procedures}.

\end{remark}

The rest of this section is devoted to proving Lemmas~\ref{lm:y-bound}, ~\ref{lm:x-bound}, ~\ref{lm:xleqy} and ~\ref{lm:xy-powerlaw}.
The analysis uses the approach of~\cite{BFNPSW14}. It turns out, however, that new ingredients are necessary for handling cycles due to the presence of multiple intermediate nodes.

\subsection{Useful facts}\label{sec:facts}
In this section we state some simple claims and known results that will be useful in the rest of the proof. The following claim specifies the probability that a fixed path exists in a random graph drawn from the Chung-Lu distribution:

\begin{claim}\label{cl:prob-path}
Let $G$ be drawn from the Chung-Lu distribution with degree sequence $\d$. Then for any $q\geq 2$ and a vector $(u_1,\ldots, u_q)\in V^q$ of distinct nodes  
$$
\prob[(u_1,\ldots, u_q)\text{~is a path in~}G]=\frac{d_{u_1} d_{u_q}}{2m}\cdot \prod_{j=2}^{q-1} \frac{d_{u_j}^2}{2m}.
$$
\end{claim}
\begin{proof}
Since the input graph $G$ is drawn from the Chung-Lu distribution, for each $j=1,\ldots, q-1$ we have 
$$
\prob[(u_j, u_{j+1})\in E]=\frac{d_{u_j} d_{u_{j+1}}}{2m}.
$$
Since edges are included in $E$ independently, we have 
\begin{equation*}
\prob[(u_1,\ldots, u_q)\text{~is a path in~}G] =\prod_{j=1}^{q-1} \frac{d_{u_j} d_{u_{j+1}}}{2m}
=\frac{d_{u_1} d_{u_q}}{2m}\cdot \prod_{j=2}^{q-1} \frac{d_{u_j}^2}{2m}
\end{equation*}
\end{proof}

We will also use 
\begin{theorem}[Chernoff bound]\label{thm:chernoff}
Let $X_1,\ldots, X_n$ be independent $0/1$ Bernoulli random variables. Let $X=\sum_{i=1}^n X_i$, and let $\mu:=\expect[X]$. Then for any $\e\in (0, 1)$
$$
\prob[X \leq (1-\e)\mu]<e^{-\Omega(\e^2 \mu)}
$$
and for every $\lambda>0$ 
$$
\prob[X \geq (1+\lambda)\mu]<(e/(1+\lambda))^{(1+\lambda)\mu}
$$
\end{theorem}
A simple corollary of the bound is
\begin{corollary}\label{cr:lower-vs-upper}
Let $X_1,\ldots, X_n$ be independent $0/1$ Bernoulli random variables. Let $X=\sum_{i=1}^n X_i$, and let $\mu_X:=\expect[X]$. 
Let $Y_1,\ldots, Y_n$ be independent $0/1$ Bernoulli random variables.
Let $Y=\sum_{i=1}^n Y_i$, and let $\mu_Y:=\expect[Y]$. If $\mu_Y\leq
\mu_X/20$, then
$
\prob[Y\geq X]< 2e^{-c\mu_X},
$
for some constant $c>0$.
\end{corollary}
\begin{proof}
We have 
\begin{equation*}
\begin{split}
\prob[Y\geq X]\leq \prob[X\leq (\mu_X+\mu_Y)/2]+\prob[Y\geq (\mu_X+\mu_Y)/2].
\end{split}
\end{equation*}
Since $(\mu_X+\mu_Y)/2\leq (3/4)\mu_X$, we have 
$\prob[X\leq (\mu_X+\mu_Y)/2] < e^{-c\mu_X}$, for some constant $c>0$ 
by the first bound from Theorem~\ref{thm:chernoff} invoked with $\e=1/4$.

We also have 
$
(\mu_X+\mu_Y)/2\geq 10\mu_Y,
$
so by the second bound from Theorem~\ref{thm:chernoff} invoked with $\lambda=9$ we get
$$
\prob[Y\geq (\mu_X+\mu_Y)/2] < (e/10)^{10\mu_Y} \leq (e/10)^{\mu_X/2}.
$$
Clearly, $(e/10)^{\mu_X/2}=e^{-c\mu_X}$ for some constant $c>0$, as required.
\end{proof}

Finally, we need the following inequality on the degree sequence $\d$.
\begin{claim}\label{cl:hoelder}
For all $t\geq 1$ one has
$$
\sum_{u} d_u^{2-1/t}\leq \left(\sum_{u} d_u\right)^{1/t}\cdot \left(\sum_{u} d_u^2\right)^{(t-1)/t}.
$$
\end{claim}
\begin{proof}
By H\"{o}lder's inequality with conjugates $t$ and $\frac1{1-1/t}$ we have 
\begin{equation*}
\begin{split}
\sum_{u} d_u^{2-1/t}&=\sum_{u} d_u^{1/t} \cdot d_u^{2-2/t}\\
&\leq \left(\sum_{u} (d_u^{1/t})^t\right)^{1/t}\cdot \left(\sum_{u} d_u^{(2-2/t)\cdot \frac1{1-1/t}}\right)^{1-1/t}\\
&\leq \left(\sum_{u} d_u\right)^{1/t}\cdot \left(\sum_{u} d_u^2\right)^{(t-1)/t}\\
\end{split}
\end{equation*}
\end{proof}

\subsection{Proofs of the lower and upper bounds}

We start by lower bounding $\expect[Y(q)]$.
\begin{lemma}\label{lm:y-bound}
Let $G=(V, E), V=[n]$ be drawn from the Chung-Lu distribution with degree sequence $\d$. Suppose that node id's are chosen uniformly at random. For any integer $q\geq 3$ let $Y(q)$ be defined by~\eqref{eq:y-def}. 
Then if $\d$ is $n^{-\delta}$-balanced for a constant $\delta>0$, the following holds for any constant $q$  
$$
\expect[Y(q)]\geq (1-o(1))\cdot \frac1{q}  (2m)^{-q+3}\left(\sum_{u} d_u^2\right)^{q-2}.
$$
\end{lemma}
\begin{proof}
We have
\begin{equation}\label{eq:2903hr}
\begin{split}
\expect[Y(q)]&= \sum_{\substack{(u_1, \ldots, u_q)\in V^q:\\
u_i\text{~distinct}}} \prob[(u_1,\ldots, u_q)\text{~is a path
in~}G]\cdot \prob\left[\text{id}(u_1)>\text{id}(u_j), j\in [2..q]\right]\\
\end{split}
\end{equation}

By Claim~\ref{cl:prob-path} one has for any vector $(u_1,\ldots, u_q)$ of distinct nodes
\begin{equation*}
\begin{split}
\prob[(u_1,\ldots, u_q)\text{~is a path in~}G]&=\prod_{j=1}^{q-1} \frac{d_{u_j} d_{u_{j+1}}}{2m}=\frac{d_{u_1} d_{u_q}}{2m}\cdot \prod_{j=2}^{q-1} \frac{d_{u_j}^2}{2m}.
\end{split}
\end{equation*}

Furthermore, for any fixed $q$-tuple of distinct nodes $(u_1,\ldots, u_q)$ we have 
$$
\prob\left[\text{id}(u_1)>\text{id}(u_j), j\in [2..q]\right]=\frac1{q}
$$
since id's are uniformly random by assumption of the lemma. (Note that
this implies that  $Y(q)$ is lower bounded by a
constant times the total number of paths of length $q$, since $q$ is
constant.)

Plugging this bound into~\eqref{eq:2903hr}, we get
\begin{equation}\label{eq:y}
\begin{split}
    \expect[Y(q)]&=\sum_{\substack{(u_1, \ldots, u_q)\in V^q:\\ u_i\text{~distinct}}} 
    \prob[(u_1,\ldots, u_q)\text{~is a path in~}G]\cdot 
    \prob\left[\text{id}(u_1)>\text{id}(u_j), j\in [2..q]\right]\\
    &=\frac 1{q}\sum_{\substack{(u_1, \ldots, u_q)\in V^q:\\ u_i\text{~distinct}}} 
    \frac{d_{u_1}\cdot d_{u_q}}{2m}\cdot \prod_{j=2}^{q-1} \frac{d_{u_j}^2}{2m}\\
    &=\frac 1{q} \left(\sum_{u_1\in V} d_{u_1}\right)\cdot
    \left(\sum_{u_q\in V} \frac{d_{u_q}}{2m}\right)\cdot
    \prod_{j=2}^{q-1} \sum_{u_j\in V}\frac{d_{u_j}^2}{2m} \\
    &~~~~-\frac 1{q}\sum_{\substack{(u_1, \ldots, u_q)\in V^q:\\ u_i\text{~{\bf not} distinct}}} 
    \frac {d_{u_1}\cdot d_{u_q}}{2m}\cdot \prod_{j=2}^{q-1} \frac{d_{u_j}^2}{2m}.
\end{split}
\end{equation}
Clearly, 
\begin{equation}\label{eq:lbound1}
    \frac1{q} \left(\sum_{u_1\in V} d_{u_1}\right)\cdot
    \left(\sum_{u_q\in V} \frac{d_{u_q}}{2m}\right)\cdot
    \prod_{j=2}^{q-1} \sum_{u_j\in V}\frac{d_{u_j}^2}{2m} = 
     \frac 1q (2m)^{3-q} \cdot \left( \sum_{u\in V}{d_{u}^2} \right)^{q-2}
\end{equation}

We now show that 
\begin{equation}\label{eq:s-star}
    \frac1{q}\sum_{\substack{(u_1, \ldots, u_q)\in V^q:\\ u_i\text{~{\bf not} distinct}}} 
    \frac{d_{u_1}\cdot d_{u_q}}{2m}\cdot \prod_{j=2}^{q-1} \frac{d_{u_j}^2}{2m}
    =o\left(\frac1{q}  (2m)^{3-q} \cdot\left(\sum_{u} d_u^2\right)^{q-2}\right).
\end{equation} Together with~\eqref{eq:y} this gives the result. 

We show that~\eqref{eq:s-star} follows from the assumption that the
degree sequence $\d$ is balanced. 
Let 
\[
\begin{split}
    \calU &= \left\{(u_1, \ldots, u_q)\in V^q: u_i\text{~{\bf not} distinct}\right\} \\
    \calW_{k,\ell} &= \left\{(u_1, \ldots, u_q)\in V^q: u_k=u_\ell\right\},
    \mbox{ for } 1 \le k < \ell \le q.
\end{split}
\]
Note that $\calU \subseteq \cup_{1 \le k < \ell \le q} \calW_{k,\ell}$.
Let $\vec{v}$ denote a vector $(u_1, \ldots, u_q)\in V^q$. To simplify the exposition
denote 
\[
    S(\vec{v}) =  \frac{d_{u_1}\cdot d_{u_q}}{2m}\cdot \prod_{j=2}^{q-1} \frac{d_{u_j}^2}{2m}.
\]

We show that for every $1 \le k < \ell \le q$, 
\begin{equation}\label{eq:set-bound}
\sum_{\vec{v} \in \calW_{k,\ell}} S(\vec{v}) 
    \leq \frac 1{q} n^{-\delta} (2m)^{3-q} \cdot\left(\sum_{u} d_u^2\right)^{q-2}.
\end{equation}
Since there are constant number of such sets this implies~\eqref{eq:s-star}. 
To prove~\eqref{eq:set-bound} we consider four different cases.

\noindent{\bf Case 1:}  $2 \le k < \ell \leq q-1$. In this case
\[
\begin{split}
\sum_{\vec{v} \in \calW_{k,\ell}} S(\vec{v}) 
&=\frac 1{q}\sum_{(u_1, \ldots, u_q)\in \calW_{k,\ell}} 
\frac{d_{u_1}\cdot d_{u_{q}}}{2m}\cdot 
\frac{d_{u_k}^4}{4m^2} \cdot
\prod_{j \in [2..q-1] \setminus \{k,\ell\}} \frac{d_{u_j}^2}{2m} \\
&=\frac 1{q} \left(\sum_{u_1\in V} d_{u_1}\right)\cdot
    \left(\sum_{u_q\in V} \frac{d_{u_q}}{2m}\right)\cdot
    \left(\sum_{u_k\in V} \frac{d_{u_k}^4}{4m^2}\right)\cdot
    \prod_{j \in [2..q-1] \setminus \{k,\ell\}} \sum_{u_j\in V}\frac{d_{u_j}^2}{2m} \\
    &\leq \frac 1{q} n^{-\delta} \left(\sum_{u_1\in V} d_{u_1}\right)\cdot
    \left(\sum_{u_q\in V} \frac{d_{u_q}}{2m}\right)\cdot
    \left(\sum_{u_k\in V} \frac{d_{u_k}^2}{2m}\right)\cdot
    \left(\sum_{u_k\in V} \frac{d_{u_k}^2}{2m}\right)\cdot \\
    &~~~~\prod_{j \in [2..q-1] \setminus \{k,\ell\}} \sum_{u_j\in V}\frac{d_{u_j}^2}{2m} \\
    &= \frac 1q n^{-\delta} (2m)^{3-q} \cdot \left( \sum_{u\in V}{d_{u}^2} \right)^{q-2}
\end{split}
\]
Note that the inequality follows since the degree sequence $\d$ is
$n^{-\delta}$-bounded.

\noindent{\bf Case 2:}  $k=1$ and $\ell=q$. In this case
\[
\begin{split}
\sum_{\vec{v} \in \calW_{1,q}} S(\vec{v}) 
&=\frac 1{q}\sum_{(u_1, \ldots, u_q)\in \calW_{1,q}} 
\frac{d_{u_1}^2}{2m}\cdot 
\prod_{j \in [2..q-1]} \frac{d_{u_j}^2}{2m} \\
&=\frac 1{q} 
    \left(\sum_{u_1\in V} \frac{d_{u_1}^2}{2m}\right)\cdot
    \prod_{j \in [2..q-1]} \sum_{u_j\in V}\frac{d_{u_j}^2}{2m} \\
&\leq \frac 1{q} n^{-\delta} \left(\sum_{u_1\in V} d_{u_1}\right)\cdot
    \left(\sum_{u_1\in V} \frac{d_{u_1}}{2m}\right)\cdot
    \prod_{j \in [2..q-1]} \sum_{u_j\in V}\frac{d_{u_j}^2}{2m} \\
&= \frac 1q n^{-\delta} (2m)^{3-q} \cdot \left( \sum_{u\in V}{d_{u}^2} \right)^{q-2}
\end{split}
\]

\noindent{\bf Case 3:}  $k=1$ and $\ell \in [2..q-1]$. In this case
\[
\begin{split}
\sum_{\vec{v} \in \calW_{1,\ell}} S(\vec{v}) 
&=\frac 1{q}\sum_{(u_1, \ldots, u_q)\in \calW_{1,\ell}} 
\frac{d_{u_1}^3}{2m} \cdot
\frac{d_{u_q}}{2m}\cdot 
\prod_{j \in [2..q-1] \setminus \{\ell\}} \frac{d_{u_j}^2}{2m} \\
&=\frac 1{q} 
    \left(\sum_{u_1\in V} \frac{d_{u_1}^3}{2m}\right)\cdot
    \left(\sum_{u_q\in V} \frac{d_{u_q}}{2m}\right)\cdot
    \prod_{j \in [2..q-1] \setminus \{\ell\}} \sum_{u_j\in V}\frac{d_{u_j}^2}{2m} \\
&\leq \frac 1{q} n^{-\delta} \left(\sum_{u_1\in V} d_{u_1}\right)\cdot
    \left(\sum_{u_q\in V} \frac{d_{u_q}}{2m}\right)\cdot
    \left(\sum_{u_1\in V} \frac{d_{u_1}^2}{2m}\right)\cdot
    \prod_{j \in [2..q-1] \setminus \{\ell\}} \sum_{u_j\in V}\frac{d_{u_j}^2}{2m} \\
&= \frac 1q n^{-\delta} (2m)^{3-q} \cdot \left( \sum_{u\in V}{d_{u}^2} \right)^{q-2}
\end{split}
\]
\noindent{\bf Case 4:}  $k \in [2..q-1]$ and $\ell=q$. This case is
symmetric to Case 3.
\end{proof}

The following lemma provides an upper bound on the expected runtime of the degree-based algorithm for enumerating cycles of length $q\geq 3$:

\begin{lemma}\label{lm:x-bound}
Let $G=(V, E), V=[n]$ be drawn from the Chung-Lu distribution with degree sequence $\d$. 
For any integer $q\geq 3$ let $X(q)$ be defined by~\eqref{eq:x-def}.  Then there exists an absolute constant $C>1$ such that 
$$
\expect[X(q)]\leq C (2m)^{-q+2} \left(\sum_{u\in V} d_u^{2-1/(q-1)}\right)^{q-1}.
$$
\end{lemma}
\begin{proof}
We have by~\eqref{eq:x-def}
\begin{equation*}
\begin{split}
\expect[X(q)]&\leq \sum_{\substack{(u_1, \ldots, u_q)\in V^q:\\ u_i\text{distinct}}} 
\prob[(u_1,\ldots, u_q)\text{~is a path in~}G\text{~and~}
\forall j\in [2..q]~\deg(u_1)\geq\deg(u_j)]
\end{split}
\end{equation*}
where $\deg(u)$ stands for the {\em actual} degree of the node $u\in V$ in the graph $G$. Note that while $\expect[\deg(u)]=d_u$, there may be deviations due to the sampling process. This fact introduces some complications in the analysis.

Similarly to~\cite{BFNPSW14}, we start by splitting the summation
above into two. For a constant $\phi$ (to be fixed later to $1/80$)
let
\begin{equation}\label{eq:x}
\begin{split}
\expect[X(q)]
\eat{
&\leq \sum_{\substack{(u_1, \ldots, u_q)\in V^q: \\ u_i\text{~distinct},\\ \deg(u_q)>\deg(u_j)}} \prob[(u_1,\ldots, u_q)\text{~is a path in~}G]\\
&=\sum_{\substack{(u_1, \ldots, u_q)\in V^q: \\ u_i\text{~distinct}, \\  d_{u_q}>  \delta d_{u_j}\text{~for all~}j=1,\ldots, q-1}} \prob[(u_1,\ldots, u_q)\text{~is a path in~}G\text{~and~}\deg(u_q)>\deg(u_j)\forall j\in [1:q-1]\\
&+\sum_{\substack{(u_1, \ldots, u_q)\in V^q: u_i\text{~distinct} \\ \exists j\in [1:q-1] \text{~s.t.~}  d_{u_q}\leq  \delta d_{u_j}}} \prob[(u_1,\ldots, u_q)\text{~is a path in~}G\text{~and~}\deg(u_q)>\deg(u_j)\forall j\in [1:q-1]\\
}
=S_1+S_2,
\end{split}
\end{equation}
where 
\begin{equation*}
\begin{split}
S_1&:=\sum_{\substack{(u_1, \ldots, u_q)\in V^q: u_i\text{~distinct},\\
\forall j\in [2..q]~d_{u_1}>  \phi d_{u_j}}}
\prob[(u_1,\ldots, u_q)\text{~is a path in~}G
\text{~and~}\forall j\in [2..q]~\deg(u_1)\geq\deg(u_j)] \\
S_2&:=\sum_{\substack{(u_1, \ldots, u_q)\in V^q: u_i\text{~distinct} \\ 
\exists j\in [2..q] \text{~s.t.~}  d_{u_1}\leq  \phi d_{u_j}}} 
\prob[(u_1,\ldots, u_q)\text{~is a path in~}G
\text{~and~}\forall j\in [2..q]~\deg(u_1)\geq\deg(u_j)].
\end{split}
\end{equation*}

We now bound the two summations separately.

\paragraph{Bounding $S_1$} 
To bound $S_1$ we use the fact that for any $q$-tuple of distinct
$u_1,\ldots, u_q$ by Claim~\ref{cl:prob-path}
\begin{equation*}
\begin{split}
&\prob[(u_1,\ldots, u_q)\text{~is a path in~}G\text{~and~}\forall j\in [2..q]~\deg(u_1)\geq\deg(u_j)\\
&\leq \prob[(u_1,\ldots, u_q)\text{~is a path in~}G]
=\frac{d_{u_1}\cdot d_{u_q}}{2m}\cdot \prod_{j=2}^{q-1} \frac{d_{u_j}^2}{2m}.
\end{split}
\end{equation*}
We thus get
\begin{equation}\label{eq:s1}
\begin{split}
S_1&\leq \sum_{\substack{(u_1, \ldots, u_q)\in V^q\\\forall j\in [2..q]~d_{u_1}>  \phi d_{u_j}}} 
    \frac{d_{u_1}\cdot d_{u_q}}{2m}\cdot \prod_{j=2}^{q-1} \frac{d_{u_j}^2}{2m}
\end{split}
\end{equation}
Since $\phi d_{u_j}< d_{u_1}$ for all potential paths $(u_1,\ldots, u_q)$ in the summation above, we have
\begin{equation*}
\begin{split}
    d_{u_1}^{\left( 1-\frac 1{q-1}\right)}&=d_{u_1}^{\frac{q-2}{q-1}} 
    > \prod_{j=2}^{q-1} (\phi d_{u_j})^\frac 1{q-1}
    = \phi^{\frac{q-2}{q-1}}\prod_{j=2}^{q-1} d_{u_j}^\frac 1{q-1}
    \geq \phi \prod_{j=2}^{q-1} d_{u_j}^\frac 1{q-1}.
\end{split}
\end{equation*}
It follows that 
\begin{equation*}
\begin{split}
d_{u_1}\cdot \prod_{j=2}^{q-1} d_{u_j}^2
  = d_{u_1}\cdot \prod_{j=2}^{q-1} d_{u_j}^{\left(2-\frac 1{q-1}\right)}
  \cdot \prod_{j=2}^{q-1} d_{u_j}^{\frac 1{q-1}}
  \leq \frac 1\phi \prod_{j=1}^{q-1} d_{u_j}^{\left(2-\frac 1{q-1}\right)}.
\end{split}
\end{equation*}

Substituting this bound in~\eqref{eq:s1}, we get 
\begin{align*}
S_1&\leq \sum_{\substack{(u_1, \ldots, u_q)\in V^q\\\forall j\in [2..q]~d_{u_1}>  \phi d_{u_j}}} 
    \frac{d_{u_1}\cdot d_{u_q}}{2m}\cdot \prod_{j=2}^{q-1} \frac{d_{u_j}^2}{2m}
    &\leq \frac 1\phi 
         \sum_{\substack{(u_1, \ldots, u_q)\in V^q\\\forall j\in [2..q]~d_{u_1}>  \phi d_{u_j}}} 
             d_{u_q}\cdot \prod_{j=1}^{q-1} \frac{d_{u_j}^{2-1/(q-1)}}{2m}\\
    &\leq \frac 1\phi 
         \sum_{\substack{(u_1, \ldots, u_q)\in V^q}}
             d_{u_q}\cdot \prod_{j=1}^{q-1} \frac{d_{u_j}^{2-1/(q-1)}}{2m}
    &\leq \frac 1\phi
         \left(\sum_{u_q\in V} d_{u_q}\right)\cdot 
         \prod_{j=1}^{q-1} \sum_{u_j\in V} \frac{d_{u_j}^{2-1/(q-1)}}{2m}\\
         &\leq  \frac 1\phi (2m)^{-q+2}
         \left(\sum_{u} d_u^{\left(2-\frac 1{q-1}\right)}\right)^{q-1}.
\end{align*}

\paragraph{Bounding $S_2$} Recall that 
\begin{equation*}
\begin{split}
    S_2&=\sum_{\substack{(u_1, \ldots, u_q)\in V^q: \\ 
    \exists j\in [2..q] \text{~s.t.~}  d_{u_1}\leq  \phi d_{u_j}}} 
    \prob\left[(u_1,\ldots, u_q)\text{~is a path in~}G\text{~and~}
    \forall j\in [2..q]~\deg(u_1)\geq\deg(u_j)\right]\\
       &=\sum_{\substack{(u_1, \ldots, u_q)\in V^q: \\ 
    \exists j\in [2..q] \text{~s.t.~}  d_{u_1}\leq  \phi d_{u_j}}}
    \prob[\forall j\in [2..q]~\deg(u_1)\geq\deg(u_j)\mid 
    \E(u_1,\ldots, u_q)]\cdot \prob[\E(u_1,\ldots, u_q)],\\
\end{split}
\end{equation*}
where we let $\E(u_1,\ldots, u_q):=\{(u_1,\ldots, u_q)\text{~is a path in~}G\}$. We omit the argument of $\E$ below to simplify notation. Note that we have 
\begin{equation*}
\begin{split}
d_{u_1}\leq &\expect[\deg(u_1)\mid\E]\leq d_{u_1}+1\\
    d_{u_j}\leq &\expect[\deg(u_j)\mid\E]\leq d_{u_j}+2\text{~for all~}j\in [2..q-1]\\
d_{u_q}\leq &\expect[\deg(u_q)\mid\E]\leq d_{u_q}+1.\\
\end{split}
\end{equation*}
Furthermore, $\deg(u_j)$ is a sum of independent $0/1$ Bernoulli
random variables even conditional on the event $\E$. Now we would like
to apply Corollary~\ref{cr:lower-vs-upper} to random variables
$\deg(u_1)$ and $\deg(u_j)$ conditional on $\E$, but there is one more
issue: these random variables are dependent through the potential edge
$(u_1, u_j)$. To avoid this issue, we omit the $0/1$ Bernoulli
random variable corresponding to this potential edge from both
random variables $\deg(u_1)$ and $\deg(u_j)$.
Let $\hdeg(u_1)$ and  $\hdeg(u_j)$ be the modified random
variables. Namely, $\hdeg(u_1):=\deg(u_1)-\mathbf{1}_{(u_1, u_j)\in E}$ 
and $\hdeg(u_j):=\deg(u_q)-\mathbf{1}_{(u_1, u_j)\in E}$, where
$\mathbf{1}_{(u_1, u_j)\in E}$ is the random variable corresponding to
the sampling of the potential edge $(u_1, u_j)$. 

Note that conditional on $\E$ the random variables $\hdeg(u_1)$
and $\hdeg(u_j)$ are independent, and are both sums of independent Bernoulli $0/1$ random variables. 
Note that
$d_u\geq 1$ for every node $u\in V$, and since 
$d_{u_1}\leq  \phi d_{u_j}$, then 
$d_{u_j}\geq 1/\phi$. It follows that 
\begin{equation}\label{eq:111a-def}
\begin{split}
\expect[\hdeg(u_1)\mid\E]&=\expect[\deg(u_1)-\mathbf{1}_{(u_1, u_j)\in E}\mid\E]
    \leq  d_{u_1}+1 \leq \phi (d_{u_j}+\frac 1\phi)\leq 2\phi d_{u_j}
\end{split}
\end{equation}
On the other hand since $d_{u_1} <m$, 
\begin{equation}\label{eq:111b-def}
\begin{split}
\expect[\hdeg(u_j)\mid\E]&=\expect[\deg(u_j)-\mathbf{1}_{(u_1, u_j)\in E}\mid\E]
    \geq  d_{u_j}\left(1-\frac{d_{u_1}}{2m}\right) \geq \half d_{u_j}.\\
\end{split}
\end{equation}

Putting these bounds together with the assumption that $\phi<1/80$, we get that the preconditions of Corollary~\ref{cr:lower-vs-upper} are satisfied, and hence
\begin{equation}\label{eq:bound}
\begin{split}
\prob[\deg(u_1)\geq \deg(u_j)\mid\E]
&=\prob[\deg(u_1)-\mathbf{1}_{(u_1, u_j)\in E}\geq \deg(u_j)-\mathbf{1}_{(u_1, u_j)\in E}\mid\E]\\
&=\prob[\hdeg(u_1)\geq \hdeg(u_j)\mid\E]\leq 2e^{-\Omega(d_{u_j})}.
\end{split}
\end{equation} 
This allows us to bound 
$\prob[\forall j\in [2..q]~\deg(u_1)\geq\deg(u_j)\mid \E]$ as follows. 
\eat{
    Let ${\mathcal I}\subseteq [2..q-1]$ 
denote the set of indices $i$ such that 
$d_{u_q}\leq \delta d_{u_i}$\footnote{Note that we could strengthen the argument by defining $\mathcal{I}$ as  the (potentially larger) set of indices $i\in [1:q-1]$ with the same property, but the current definition is sufficient for our purposes.}.
}
Let $j^*$ be the index of the highest expected degree node in 
$\{u_2,\ldots,u_q\}$, that is, $j^*:=\text{argmax}_{j\in [2..q]} d_{u_j}$. Clearly,
\begin{equation*}
\begin{split}
\prob[\forall j\in [2..q]~\deg(u_1)\geq\deg(u_j)\mid \E]
     &\leq \prob[\deg(u_1)\geq\deg(u_{j^*})\mid \E].
\end{split}
\end{equation*}
Note that if there exists $j\in [2..q]$ such that $d_{u_1}\leq  \phi d_{u_j}$, 
then also $d_{u_1}\leq  \phi d_{u_{j^*}}$. 
Thus, applying~\eqref{eq:bound} with $j=j^*$ we get
\begin{equation*}
\begin{split}
\prob[\forall j\in [2..q]~\deg(u_1)\geq\deg(u_j)\mid \E]
&\leq \prob[\hdeg(u_1)\geq \hdeg(u_{j^*})\mid\E]\\
&\leq 2e^{-\Omega(d_{u_{j^*}})}.
\end{split}
\end{equation*}

Substituting this bound in the expression for $S_2$ and using the fact that  
$2e^{-\Omega(d_{u_{j^*}})}\leq 2\prod_{j\in [1..q]} e^{-\Omega\left(\frac1{q}d_{u_j}\right)}$ 
(since $d_{u_{j^*}} \geq d_{u_j}$ for all $j\in [1..q]$), we get
\begin{equation*}
\begin{split}
    S_2 &\leq\sum_{\substack{(u_1, \ldots, u_q)\in V^q: \\ 
    \exists j\in [2..q] \text{~s.t.~}  d_{u_1}\leq  \phi d_{u_j}}} 
    (2m)^{-q+1} \cdot
    d_{u_1}e^{-\Omega(\frac1{q}d_{u_1})}\cdot
    d_{u_q}e^{-\Omega(\frac1{q}d_{u_q})}\cdot 
    \prod_{j\in [2..q-1]} d_{u_j}^2 e^{-\Omega(\frac1{q}d_{u_j})} \\
    &\leq (2m)^{-q+1} \cdot
    \left( \sum_{u \in V}  d_{u}e^{-\Omega(\frac1{q}d_{u})} \right)^2\cdot
    \left( \sum_{u \in V}  d_{u}^2 e^{-\Omega(\frac1{q}d_{u})} \right)^{q-2}.
\end{split}
\end{equation*}

Recall that the first two moments of the exponential distribution are
constants and thus since $q$ is assumed to be constant we get that
both
\[
    \sum_{u \in V}  d_{u} e^{-\Omega(\frac1{q}d_{u})}~\text{ and }~
    \sum_{u \in V}  d_{u}^2 e^{-\Omega(\frac1{q}d_{u})}
\]
are constants and thus $S_2 = O\left((2m)^{-q+1}\right)=o(S_1)$.
\end{proof}

\subsection{Comparing the lower and upper bounds}
We show that the bound on $\expect[X(q)]$ is not much worse than our
bound on $\expect[Y(q)]$ {\em for any degree sequence} $\d$ that is
$n^{-\delta}$-balanced.
We later show that our bound gives polynomially smaller runtime if the
degree sequence $\d$ satisfies the truncated power law.

\begin{lemma}\label{lm:xleqy}
For any $q\geq 3$ and any $n^{-\delta}$-balanced degree sequence $\d$, for a constant
$\delta>0$, $\expect[X(q)]=O(\expect[Y(q)])$.
\end{lemma}
\begin{proof}
We start with the bound from Lemma~\ref{lm:x-bound}:
$$
\expect[X(q)]\leq C (2m)^{2-q} \left(\sum_{u\in V} d_u^{\left(2-\frac 1{q-1}\right)}\right)^{q-1}.
$$

By Claim~\ref{cl:hoelder} with $t=q-1$ we have 
\begin{equation*}
\begin{split}
\sum_{u} d_u^{\left(2-\frac 1{q-1}\right)}&\leq \left(\sum_{u} d_u\right)^{\frac1{q-1}}\cdot \left(\sum_{u} d_u^2\right)^{\frac{q-1-1}{q-1}}.\\
\end{split}
\end{equation*}

Substituting this into the bound above, we get that 
\begin{equation*}
\begin{split}
\expect[X(q)]&\leq C (2m)^{2-q} \left(\sum_{u\in V} d_u^{\left(2-\frac 1{q-1}\right)}\right)^{q-1}\\
&\leq C (2m)^{2-q} \left(\left(\sum_{u} d_u\right)^{\frac1{q-1}}\cdot 
\left(\sum_{u} d_u^2\right)^{\frac{q-2}{q-1}}\right)^{q-1}\\
&= C (2m)^{2-q} \left(\sum_{u} d_u\right)\cdot \left(\sum_{u} d_u^2\right)^{q-2}
= C (2m)^{3-q} \left(\sum_{u} d_u^2\right)^{q-2}.
\end{split}
\end{equation*}
The lemma follows by comparing the above bound to the bound on $\expect[Y(q)]$
provided by Lemma~\ref{lm:y-bound} (note that the preconditions are
satisfied, as $\d$ is $n^{-\delta}$-balanced for a constant $\delta>0$
by assumption of the lemma).
\end{proof}

We now compare the bounds from Lemma~\ref{lm:x-bound} and
Lemma~\ref{lm:y-bound} when the graph $G$ is drawn
from the Chung-Lu distribution with a degree sequence that satisfies
the truncated power law. 
\begin{lemma}\label{lm:xy-powerlaw}
Let $G$ be a random graph drawn from the Chung-Lu distribution with degree
sequence $\d$ that satisfies
the truncated power law  with exponent $\alpha$, for a constant
$\alpha\in (1, 2)$. 
The following bounds hold for any constant $q\geq 3$. 

{\bf (1)} $\expect[Y(q)]= \Omega\left(n^{\alpha-1+\frac1{2}(2-\alpha)q}\right)$

{\bf (2a)}  $\expect[X(q)] = O\left(n^{\half+\half(2-\alpha)(q-1)}\right)$,
for $\alpha\in (1, 2-\frac 1{q-1})$.

{\bf (2b)} $\expect[X(q)]=O\left(n\log n\right)$,
for $\alpha\in [2-\frac 1{q-1},2)$.
\end{lemma}

\begin{proof} We first prove the bound on $\expect[Y(q)]$. As shown
later in Claim~\ref{cl:balanced} the condition that the degree
sequence $\d$ satisfies
the truncated power law implies that the degree sequence $\d$ is 
$n^{-\delta}$-balanced for a constant $\delta>0$, 
and hence preconditions of Lemma~\ref{lm:y-bound} are satisfied.
To apply the bound of Lemma~\ref{lm:y-bound} we bound $\sum_u d_u^2$.
\[
\sum_{u} d_u^2=
\sum_{j=0}^{\frac1{2}\log_2 n} \left(\frac{n}{2^{\alpha j}}\right) 2^{2j}
=n\cdot \sum_{j=0}^{\frac1{2}\log_2 n} 2^{(2-\alpha)j}
= \Theta\left(n\cdot n^{1-\half \alpha}\right).
\]

The number of edges in a graph with degree sequence that satisfies the
power law with exponent $\alpha \in (1,2)$ satisfies 
\begin{equation*}
    m=\half\sum_{u\in V} d_u=\Theta\left(\sum_{j\geq 0} 2^j\cdot \frac{n}{2^{\alpha j}}\right)
     =\Theta\left(\sum_{j\geq 0} 2^{(1-\alpha)j}\cdot n\right)=\Theta(n).
\end{equation*}

Plugging both bounds in the bound of
Lemma~\ref{lm:y-bound} we get
$$
\expect[Y(q)] = \Omega\left(n^{-q+3}  n^{(2-\half \alpha)(q-2)}\right)
= \Omega\left(n^{\alpha-1+\half (2-\alpha)q}\right).
$$

To bound $\expect[X(q)]$ using Lemma~\ref{lm:x-bound} we need first to
bound the following summation.
\begin{equation}\label{eq:sumx}
\begin{split}
    \sum_{u\in V} d_u^{\left(2-\frac 1{q-1}\right)}&
    =\sum_{j=0}^{\half \log_2 n} \frac{n}{2^{\alpha j}} \cdot 2^{\left(2-\frac 1{q-1}\right)j}
    =n\cdot \sum_{j=0}^{\half \log_2 n}  2^{\left(2-\alpha-\frac 1{q-1}\right)j}.
\end{split}
\end{equation}

The sum~\eqref{eq:sumx} above is dominated either by the first term
or by the last term, depending on whether $\alpha$ is less than or
greater than $2-\frac 1{q-1}$. 

\noindent{\bf Case 1:} $\alpha\in \left(1, 2-\frac 1{q-1}\right)$. In
this case the sum~\eqref{eq:sumx} is dominated by the last term and
thus bounded by
$$
n\cdot \sum_{j=0}^{\half \log_2 n}  2^{\left(2-\alpha-\frac 1{q-1}\right)j}
=O\left(n^{1+\half \left(2-\alpha-\frac 1{q-1}\right)}\right) 
=O\left(n^{2-\half \left(\alpha+\frac 1{q-1}\right)}\right). 
$$

Substituting this into the bound provided by Lemma~\ref{lm:x-bound},
using the fact that $m=\Theta(n)$ we get
\begin{equation*}
\begin{split}
\expect[X(q)]&
    = O\left(n^{-q+2} n^{\left( 2-\half \left(\alpha+\frac 1{q-1}\right)(q-1)\right)}\right)
    = O\left(n^{\half+\half(2-\alpha)(q-1)}\right)
\end{split}
\end{equation*}


\noindent{\bf Case 2:} $\alpha\in \left[2-\frac 1{q-1},2\right)$. In
this case the sum~\eqref{eq:sumx} is dominated by the first term which is constant and
since we have $\half \log n$ summands the sum~~\eqref{eq:sumx} is
$O(n\log n)$.
\eat{
$$
n\cdot \sum_{j=0}^{\frac1{2}\log_2 n} 2^{(2-1/(q-1)-\alpha)j}=O\left(n\log n\right).
$$
}

Substituting this into the bound provided by Lemma~\ref{lm:x-bound},
using the fact that $m=\Theta(n)$ we get
\begin{equation*}
\begin{split}
\expect[X(q)]&
    = O\left(n^{-q+2} \left( n\log n\right)^{q-1}\right)
    = O\left(n \left(\log n\right)^{q-1}\right).
\end{split}
\end{equation*}
\end{proof}

Applying these bounds we conclude the following 
\begin{corollary}\label{cor:xy-powerlaw}
If $G$ is a random graph drawn from the Chung-Lu distribution with degree
sequence $\d$ that satisfies
the truncated power law  with exponent $\alpha$, for a constant
$\alpha\in (1, 2)$,
then  $\expect[X(q)]= o\left(\expect[Y(q)]\right)$.
\eat{
For any constant $q$ and any $\alpha\in (1, 2)$ the expected number of paths enumerated by the degree based algorithm is smaller by a polynomial factor than the expected number of paths constructed by the id based (naive) enumeration algorithm. Specifically, for $\alpha\in (1, 2-1/(q-1))$ one has $\expect[Y(q)]/\expect[X(q)]=\Omega(n^{\frac1{2}(3-\alpha)})$, while for $\alpha\in [2-1/(q-1), 2)$ one has 
$\expect[Y(q)]/\expect[X(q)]=\Omega\left(\frac1{\log n}n^{\frac1{2}(2-\alpha)q}\right)$.
}
\end{corollary}

\begin{proof}
Using the bounds of Lemma~\ref{lm:xy-powerlaw} above. 
If $\alpha\in \left(1, 2-\frac 1{q-1}\right)$
the improvement of the degree based algorithm over the id based (naive) algorithm is 
\begin{equation*}
\begin{split}
    \frac{\expect[Y(q)]}{\expect[X(q)]}
    &\geq \frac{n^{\alpha-1+\half (2-\alpha)q}}
    {n^{\half+\half(2-\alpha)(q-1)}}
    =n^{\half(\alpha-1)}
\end{split}
\end{equation*}
If $\alpha\in \left[2-\frac 1{q-1},1\right)$
the improvement of the degree based algorithm over the id based (naive) algorithm is 
\begin{equation*}
\begin{split}
    \frac{\expect[Y(q)]}{\expect[X(q)]}
    &\geq \frac{n^{\alpha-1+\half (2-\alpha)q}}
    {n \left(\log n\right)^{q-1}}
    =n^{\alpha-2 +\half(2-\alpha)q}\cdot  \left(\log n\right)^{1-q}
\end{split}
\end{equation*}
\end{proof}

Note that if If $\alpha\in \left(2-\frac 1{q-1},1\right)$ then the
second term vanishes.

\section{Power law and balanced sequences}
In this section we show that a degree sequence that satisfies the
truncated power law is also balanced.

\begin{claim}\label{cl:balanced}
    Any degree sequence $\d$ that satisfies the truncated power law
    with exponent $\alpha$, for any $\alpha\in (1, 2)$ that is bounded
    away from $1$ and $2$ by constants, 
    is $\lambda$-balanced for $\lambda=O\left(n^{\half \alpha-1}\right)$.
\end{claim}

\begin{proof}
The intuition behind the proof is simple: the bound that needs to be
satisfied by a balanced sequence
holds because a degree sequence that
satisfies the truncated power law contains about $n^{\half(1-\alpha)}$
nodes of degree $\sqrt{n}$ (the largest), which means that the edge mass
is somewhat spread among these nodes, leading to the result of the lemma. We now give the details.

Recall that by definition of the truncated power law, for  each
$j=0,\ldots, \half \log_2 n$, the number of nodes with degree
$\Theta(2^j)$ is $\Theta(n/2^{\alpha j})$.
For any integer $s\geq 2$, we have 
\begin{equation*}
\begin{split}
\sum_{u} d_u^s&
    =\Theta\left(\sum_{j=0}^{\frac1{2}\log_2 n} 2^{s\cdot j}\cdot\frac{n}{2^{\alpha j}}\right)
    =\Theta\left(\sum_{j=0}^{\frac1{2}\log_2 n} 2^{(s-\alpha)\cdot j}\cdot n\right).
\end{split}
\end{equation*}
Since $s\geq 2$ and $\alpha$ is bounded away from $2$ by assumption,
the summation is dominated by the last term. Namely,
\begin{equation*}
\begin{split}
\sum_{u} d_u^s&
    =\Theta\left(\sum_{j=0}^{\half\log_2 n} 2^{(s-\alpha)\cdot j}\cdot n\right)
    =\Theta\left(n^{s/2}\cdot n^{1-\alpha/2}\right)
\end{split}
\end{equation*}

\eat{
We now apply the derivation above with $q=a+b$. Note that $q\geq 2$, as required, since $a, b\geq 1$. This gives 
\begin{equation}\label{eq:apb}
\begin{split}
\sum_{u} d_u^{a+b}&=\Theta(n^{(a+b)/2}\cdot n^{1-\alpha/2}).\\
\end{split}
\end{equation}
}

We now show that for 
for any integers $a, b\geq 1$,
$\sum_{u} d_u^{a+b}\leq \lambda \cdot (\sum_{u} d_u^a)(\sum_{u} d_u^b)$.

We distinguish  the following three cases:\\

\noindent {\bf Case 1: $a\geq 2$ and $b\geq 2$.} In this case the derivation above with $q=a$ and $q=b$ gives 
\begin{equation*}
\begin{split}
\sum_{u} d_u^a&=\Theta(n^{a/2}\cdot n^{1-\alpha/2})~~~~
\sum_{u} d_u^b=\Theta(n^{b/2}\cdot n^{1-\alpha/2}).\\
\end{split}
\end{equation*}
It follows that

\begin{equation*}
\begin{split}
\frac{\sum_{u} d_u^{a+b}}{(\sum_{u} d_u^a)(\sum_{u} d_u^b)}
&=\Theta\left(\frac{n^{(a+b)/2}\cdot n^{1-\alpha/2}}{n^{a/2}\cdot n^{1-\alpha/2}\cdot n^{b/2}\cdot n^{1-\alpha/2}}\right)\\
&=\Theta\left(n^{\alpha/2-1}\right),
\end{split}
\end{equation*}
which gives the result.

\noindent {\bf Case 2: $a=1$, $b\geq 2$.} Since $\alpha$ is bounded
away from $1$ by a constant we have 
\begin{equation}\label{eq:num-edges}
    \sum_{u} d_u=\sum_{j=0}^{\frac1{2}\log_2 n} 2^j\cdot \frac{n}{2^{\alpha j}}=\Theta(n)
\end{equation}
On the other hand, since $b\geq 2$ and $a+b\geq 2$, 
\begin{equation*}
\begin{split}
\sum_{u} d_u^{a+b}&=\Theta(n^{(a+b)/2}\cdot n^{1-\alpha/2})~~~~
\sum_{u} d_u^{b}=\Theta(n^{b/2}\cdot n^{1-\alpha/2})\\
\end{split}
\end{equation*}
Putting the estimates above together, we get
\begin{equation*}
\begin{split}
\frac{\sum_{u} d_u^{a+b}}{(\sum_{u} d_u^a)(\sum_{u} d_u^b)}
&=\Theta\left(\frac{n^{(a+b)/2}\cdot n^{1-\alpha/2}}{n\cdot n^{b/2}\cdot n^{1-\alpha/2}}\right)
=\Theta\left(n^{-1/2}\right).
\end{split}
\end{equation*}
The result follows since $n^{\alpha/2-1}\geq n^{-1/2}$ by the assumption that $\alpha\in (1, 2)$.
Note that the case $a \geq 2$, $b=1$ is symmetric.

\noindent {\bf Case 3: $a=1$, $b=1$.} Then we have 
\begin{equation*}
\begin{split}
\sum_{u} d_u^{a+b}&=\Theta(n^{(a+b)/2}\cdot n^{1-\alpha/2})~~~~
\sum_{u} d_u^{a}=\Theta(n)~~~~
\sum_{u} d_u^{b}=\Theta(n)
\end{split}
\end{equation*}
by the estimates above, and hence 
\begin{equation*}
\begin{split}
\frac{\sum_{u} d_u^{a+b}}{(\sum_{u} d_u^a)(\sum_{u} d_u^b)}&=\Theta\left(\frac{n^{(a+b)/2}\cdot n^{1-\alpha/2}}{n\cdot n}\right)\\
&=\Theta\left(n^{-\alpha/2}\right)
\end{split}
\end{equation*}
The result follows since $n^{\alpha/2-1}\geq n^{-\alpha/2}$ by the assumption that $\alpha\in (1, 2)$.
\end{proof}


\eat{

}

\end{document}